\documentclass[11pt]{article}
\usepackage{fullpage}
\usepackage[utf8]{inputenc}
\usepackage{algorithm}
\usepackage{algorithmic}
\usepackage{amsfonts}
\usepackage{graphicx}
\usepackage{amsthm}
\usepackage{mathrsfs}
\usepackage{amssymb}
\usepackage{hyperref}
\usepackage{amsmath}
\usepackage{cite}
\usepackage{xcolor}
\usepackage{comment}
\usepackage{mathtools}

\newtheorem{requirement}{}
\theoremstyle{plain}
\newtheorem{theorem}{Theorem}[section]
\newtheorem{lemma}[theorem]{Lemma}
\newtheorem{corollary}[theorem]{Corollary}
\theoremstyle{definition}

\newcommand{\gio}{G_\mathsf{I,O}}
\newcommand{\true}{\mathtt{True}}
\newcommand{\false}{\mathtt{False}}
\newcommand{\gi}{G_\mathsf{I}}
\newcommand{\pif}{\langle \Pi,f\rangle}
\newcommand{\msfi}{\mathsf{I}}
\newcommand{\msfo}{\mathsf{O}}

\begin{document}
\title{Twenty-Two New Approximate Proof Labeling Schemes \\ (Full Version)
}
\author{Yuval Emek\footnote{Technion - Israel Institute of Technology. yemek@technion.ac.il} \and Yuval Gil\footnote{Technion - Israel Institute of Technology. yuval.gil@campus.technion.ac.il}
}
\date{}
\maketitle
		
	\begin{abstract}
		Introduced by Korman, Kutten, and Peleg (Distributed Computing 2005), a
		\emph{proof labeling scheme (PLS)} is a system dedicated to verifying that a
		given configuration graph satisfies a certain property.
		It is composed of a centralized \emph{prover}, whose role is to generate a
		proof for yes-instances in the form of an assignment of labels to the nodes,
		and a distributed \emph{verifier}, whose role is to verify the validity of the
		proof by local means and accept it if and only if the property is satisfied.
		To overcome lower bounds on the label size of PLSs for certain graph
		properties, Censor-Hillel, Paz, and Perry (SIROCCO 2017) introduced the notion
		of an \emph{approximate proof labeling scheme (APLS)} that allows the verifier
		to accept also some no-instances as long as they are not ``too far'' from
		satisfying the property.
		
		The goal of the current paper is to advance our understanding of the power and
		limitations of APLSs.
		To this end, we formulate the notion of APLSs in terms of \emph{distributed
			graph optimization problems (OptDGPs)} and develop two generic methods for the
		design of APLSs.
		These methods are then applied to various classic OptDGPs, obtaining
		twenty-two new APLSs.
		An appealing characteristic of our APLSs is that they are all
		\emph{sequentially efficient} in the sense that both the prover and the
		verifier are required to run in (sequential) polynomial time.
		On the negative side, we establish ``combinatorial'' lower bounds on the label
		size for some of the aforementioned OptDGPs that demonstrate the optimality of
		our corresponding APLSs.
		For other OptDGPs, we establish conditional lower bounds that exploit the
		sequential efficiency of the verifier alone (under the assumption that
		$\textrm{NP} \neq \textrm{co-NP}$) or that of both the verifier and the prover
		(under the assumption that
		$\textrm{P} \neq \textrm{NP}$,
		with and without the unique games conjecture).
	\end{abstract}

\section{Introduction}
\label{section:introduction}

\subsection{Model}
\label{section:model}
Consider a connected undirected graph $G=(V,E)$ and denote $n=|V|$ and $m=|E|$. 
For a node $v \in V$, we stick to the convention that $N(v) = \{ u \mid (u, v) \in E  \}$ 
denotes the set of $v$'s \emph{neighbors} in $G$. 
An edge is said to be \emph{incident} on $v$ if it connects between $v$ and one of its neighbors. 

In the realm of \emph{distributed graph algorithms}, the nodes of graph $G =
(V, E)$ are associated with processing units that operate in a decentralized
fashion.
We assume that node $v\in V$ distinguishes between its incident edges by means
of \emph{port numbers}, i.e.,  a bijection between the set of edges incident
on $v$ and the integers in
$\{1,...,|N(v)|\}$.
Additional graph attributes, such as node ids, edge orientation, and edge and
node weights, are passed to the nodes by means of an \emph{input assignment}
$\mathsf{I} : V \rightarrow \{ 0, 1 \}^{*}$ that assigns to each node
$v \in V$,
a bit string $\mathsf{I}(v)$, referred to as $v$'s \emph{local input}, that
encodes the additional attributes of $v$ and its incident edges.
The nodes return their output by means of an \emph{output assignment}
$\mathsf{O} : V \rightarrow \{ 0, 1 \}^{*}$
that assigns to each node
$v \in V$,
a bit string $\mathsf{O}(v)$, referred to as $v$'s \emph{local output}.
We often denote
$\gio = \langle G,\mathsf{I, O}\rangle$ and $\gi=\langle G,\mathsf{I}\rangle$
and refer to these tuples as an \emph{input-output (IO) graph}  and an
\emph{input graph}, respectively.\footnote{Refer to Table \ref{abbreviations} for a full list of the abbreviations used in this paper.}

A \emph{distributed graph problem (DGP)} $\Pi$ is a collection of IO graphs
$G_\mathsf{I, O}$.
In the context of a DGP $\Pi$, an input graph $\gi$ is said to be \emph{legal} (and the graph $G$ and input assignment $\mathsf{I}$ are said to be \emph{co-legal}) if there exists an output assignment $\mathsf{O}$ such that $\gio\in \Pi$, in which case we say that $\mathsf{O}$ is a \emph{feasible solution} for $\gi$ (or simply for $G$ and $\mathsf{I}$).
Given a DGP $\Pi$, we may slightly abuse the notation and write $\gi\in \Pi$ to denote that $\gi$ is legal. 

A \emph{distributed graph minimization problem (MinDGP)} (resp.,
\emph{distributed graph maximization problem (MaxDGP)}) $\Psi$ is a pair
$\langle \Pi , f\rangle $,
where $\Pi$ is a DGP and
$f:\Pi \rightarrow \mathbb{Z}$
is a function, referred to as the \emph{objective function} of $\Psi$, that
maps each IO graph
$G_\mathsf{I, O} \in \Pi$
to an integer value
$f(G_\mathsf{I, O})$.\footnote{%
	 We assume for simplicity that the images of the objective functions used
	in the context of this paper, are integral.
	Lifting this assumption and allowing for real numerical values would
	complicate some of the arguments, but it does not affect the validity of our
	results.}
Given a co-legal graph $G$ and input assignment $\mathsf{I}$, define
\[
OPT_{\Psi}(G, \mathsf{I}) = \underset{\mathsf{O}:G_\mathsf{I,O}\in \Pi}{\inf} \{f(G_\mathsf{I, O})\}
\]
if $\Psi$ is a MinDGP; and 
\[
OPT_{\Psi}(G, \mathsf{I})= \underset{\mathsf{O}:G_\mathsf{I,O}\in \Pi}{\sup} \{f(G_\mathsf{I, O})\}
\]
if $\Psi$ is a MaxDGP.
We often use the general term \emph{distributed graph optimization problem
	(OptDGP)} to refer to MinDGPs as well as MaxDGPs.
Given a OptDGP $\Psi=\langle \Pi, f\rangle$ and co-legal graph $G$ and input assignment $\mathsf{I}$, the output assignment $\mathsf{O}$ is said to be an \emph{optimal solution} for $\gi$ (or simply for $G$ and $\mathsf{I}$) if $\mathsf{ O}$ is a feasible solution for $\gi$ and $f(G_\mathsf{I, O})=OPT_\Psi(G,\mathsf{I})$. 

Let us demonstrate our definitions through the example of the maximum weight matching problem in bipartite graphs, i.e., explaining how it fits into the framework of a MaxDGP
$\Psi = \langle \Pi, f \rangle$. 
Given a graph $G=(V,E)$ and an input assignment $\mathsf{I}$, the input graph $\gi$ is legal (with respect to $\Pi$) if $G$ is bipartite and $\mathsf{I}$ encodes an edge weight function
$w : E \rightarrow \mathbb{Z}$.
Formally, for every node $v \in V$, the local input assignment $\mathsf{I}(v)$ is set to be a vector, indexed by the port numbers of $v$, defined so that if edge
$e = (u, u') \in E$
corresponds to ports
$1 \leq i \leq |N(u)|$
and
$1 \leq i' \leq |N(u')|$
at nodes $u$ and $u'$, respectively, then both
the $i$-th entry in $\mathsf{I}(u)$
and
the $i'$-th entry in $\mathsf{I}(u')$
hold the value $w(e)$.
Given a legal input graph $\gi\in \Pi$, the output assignment $\mathsf{O}$ is a feasible solution for $\gi$ if $\mathsf{O}$ encodes a matching $\mu \subseteq E$ in $G$.
Formally, the local output assignment $\mathsf{O}(v)$ is set to the port
number corresponding to $e$ if there exists an edge $e \in \mu$ incident on
$v$; and to $\bot$ otherwise. The objective function $f$ of $\Psi$ is defined
so that for an IO graph $\gio\in\Pi$ with corresponding edge weight function
$w_\mathsf{I}:E\rightarrow\mathbb{Z}$ and matching $\mu_\mathsf{O}\subseteq
E$, the value of $f$ is set to
$f(\gio)=\sum_{e\in\mu_\mathsf{O}}w_\mathsf{I}(e)$. Following this notation, a
feasible solution $\mathsf{O}$ for co-legal $G$ and $\mathsf{I}$ is optimal if
and only if $\mu_\mathsf{O}$ is a maximum weight matching in $G$ with respect
to the edge weight function $w_\mathsf{I}$. 

While the formulation introduced in the current section is necessary for the
general definitions presented in Section \ref{section:pls} and the generic
methods developed in Section~\ref{section:methods}, in
Section~\ref{section:results}, when considering IO graphs in the context of
specific DGPs and OptDGPs, we often do not explicitly describe the input and
output assignments, but rather take a more natural high-level approach.
For example, in the context of the aforementioned maximum weight matching
problem in bipartite graphs, we may address the input edge weight function and
output matching directly without providing an explanation as to how they are
encoded in the input and output assignments, respectively.
The missing details would be clear from the context and could be easily
completed by the reader.

\subsubsection{Proof Labeling Schemes}
\label{section:pls}
In this section we present the notions of proof labeling schemes \cite{pls}
and approximate proof labeling schemes \cite{aplspaper} for OptDGPs and their
decision variants.
To unify the definitions of these notions, we start by introducing the notion
of gap proof labeling schemes based on the following definition.

A \emph{configuration graph} $G_S=\langle G,S\rangle$ is a pair consisting of a graph $G=(V,E)$ and a function $S:V\rightarrow\{0,1\}^*$ assigning a bit string $S(v)$ to each node $v\in V$. In particular, an input graph $\gi$ is a configuration graph, where $S(v)=\msfi(v)$, and an IO graph $\gio$ is a configuration graph, where $S(v)=\msfi(v)\cdot\msfo(v)$. 

Fix some universe $\mathcal{U}$ of configuration graphs.
A \emph{gap proof labeling scheme (GPLS)} is a mechanism designed to
distinguish the configuration graphs
in a \emph{yes-family}
$\mathcal{F}_{Y} \subset \mathcal{U}$
from the configuration graphs
in a \emph{no-family}
$\mathcal{F}_{N} \subset \mathcal{U}$,
where
$\mathcal{F}_{Y} \cap \mathcal{F}_{N} = \emptyset$.
This is done by means of a (centralized) \emph{prover} and a (distributed)
\emph{verifier} that play the following roles:
Given a configuration graph $G_S\in \mathcal{U}$, if $G_S\in \mathcal{F}_Y$,
then the prover assigns a bit string $L(v)$, called the \emph{label} of $v$,
to each node $v\in V$.
Let $L^N(v)$ be the vector of labels assigned to $v$'s neighbors. The verifier
at node
$v \in V$
is provided with the $3$-tuple
$\langle S(v), L(v), L^{N}(v) \rangle$
and  returns a Boolean value $\varphi(v)$. 

We say that the verifier \emph{accepts} $G_S$ if
$\varphi(v) = \true$
for all nodes
$v \in V$;
and that the verifier \emph{rejects} $G_S$ if
$\varphi(v) = \false$
for at least one node
$v \in V$.
The GPLS is said to be \emph{correct} if the following requirements hold for
every configuration graph
$G_S \in \mathcal{U}$:
\begin{requirement}\label{gpls-r1}
	If $G_S\in \mathcal{F}_Y$, then the prover produces a label assignment $L : V \rightarrow \{ 0, 1 \}^{*}$ such that the verifier accepts $G_S$.
\end{requirement}
\begin{requirement}\label{gpls-r2}
	If $G_S\in \mathcal{F}_N$, then for any label assignment $L : V \rightarrow \{ 0, 1 \}^{*}$, the verifier rejects $G_S$.
\end{requirement}
We emphasize that no requirements are made for configuration graphs
$G_S \in \mathcal{U} \setminus (\mathcal{F}_Y \cup \mathcal{F}_N)$;
in particular, the verifier may either accept or reject these configuration
graphs (the same holds for configuration graphs that do not belong to the
universe $\mathcal{U}$).

The performance of a GPLS is measured by means of its \emph{proof size}
defined to be the maximum length of a label $L(v)$ assigned by the prover to
the nodes $v \in V$ assuming that 
$G_S\in  \mathcal{F}_Y$.
We say that GPLS admits a \emph{sequentially efficient prover} if for any
configuration graph
$G_S \in {F}_Y$,
the sequential runtime of the prover is polynomial in the number of bits used
to encode $G_S$;
and that it admits a \emph{sequentially efficient verifier} if the sequential
runtime of the verifier in node
$v \in V$
is polynomial in $|S(v)|$, $|L(v)|$, and
$\sum_{u \in N(v)} |L(u)|$.
The GPLS is called \emph{sequentially efficient} if both its prover and
verifier are sequentially efficient.

\paragraph*{Proof Labeling Schemes for OptDGPs.}
Consider some OptDGP $\Psi=\pif$ and let $\mathcal{U} =\{\gio \mid \gi \in \Pi \}$. A \emph{proof labeling scheme (PLS)} for $\Psi$ is defined as a GPLS over $\mathcal{U}$ by setting the yes-family to be \[\mathcal{F}_Y=\{ \gio\in\Pi\mid f(\gio)=OPT_\Psi(G,\msfi)\}\] and the no-family to be $\mathcal{F}_N=\mathcal{U}\setminus \mathcal{F}_Y$. In other words, a PLS for $\Psi$ determines for a given IO graph $\gio\in \mathcal{U}$ whether the output assignment $\mathsf{O}:V \rightarrow \{ 0, 1 \}^{*}$ is an optimal solution (which means in particular that it is a feasible solution) for the co-legal graph $G=(V,E)$ and input assignment $\mathsf{I}:V \rightarrow \{ 0, 1 \}^{*}$.

In the realm of OptDGPs, it is natural to relax the definition of a PLS so that it may also accept feasible solutions that only approximate the optimal ones.
Specifically, given an approximation parameter
$\alpha \geq 1$,
an \emph{$\alpha$-approximate proof labeling scheme ($\alpha$-APLS)} for a OptDGP
$\Psi = \langle \Pi, f \rangle$
is defined in the same way as a PLS for $\Psi$ with the sole difference that the no-family is defined by setting
\[
\mathcal{F}_{N}
\, = \,
\begin{cases}
\mathcal{U} \setminus \left\{
\gio \in \Pi \mid f(\gio) \leq \alpha \cdot OPT_{\Psi}(G, \msfi)
\right\}
\, , &
\text{if $\Psi$ is a MinDGP} \\
\mathcal{U} \setminus \left\{
\gio \in \Pi \mid f(\gio) \geq OPT_{\Psi}(G, \msfi) / \alpha
\right\}
\, , &
\text{if $\Psi$ is a MaxDGP}
\end{cases} \, .
\]

\sloppy
\paragraph*{Decision Proof Labeling Schemes for OptDGPs.}
Consider some MinDGP (resp., MaxDGP) $\Psi=\pif$ and let
$\mathcal{U}=\{\gi\mid\gi\in \Pi\}$.
A \emph{decision proof labeling scheme (DPLS)} for $\Psi$ and a parameter $k\in \mathbb{Z}$ is defined as a GPLS over $\mathcal{U}$ by setting the yes-family to be \[
\mathcal{F}_{Y}
\, = \,
\begin{cases}
\left\{
\gi \in \Pi \mid OPT_{\Psi}(G,\msfi) \geq k
\right\} \, , &
\text{if $\Psi$ is a MinDGP} \\
\left\{
\gi \in \Pi \mid OPT_{\Psi}(G, \msfi) \leq k
\right\} \, , &
\text{if $\Psi$ is a MaxDGP}
\end{cases}
\]
and the no-family to be $\mathcal{F}_N=\mathcal{U}\setminus \mathcal{F}_Y$.
In other words, given an input graph $\gi\in\Pi$, a DPLS for $\Psi$ and $k$
decides if $f(\gio)\geq k$ (resp., $f(\gio)\leq k$) for every feasible output
assignment
$\msfo : V \rightarrow \{ 0, 1\}^{*}$.
Notice that while PLSs address the task of verifying the optimality of a given
output assignment $\msfo$, that is, verifying that no output assignment admits
an objective value smaller (resp., larger) than $f(\gio)$, in DPLSs, the
output assignment $\msfo$ is not specified and the task is to verify that no
output assignment admits an objective value smaller (resp., larger) than the
parameter $k$, provided as part of the DPLS task.
\par\fussy

Similarly to PLSs, the definition of DPLS admits a natural relaxation.
Given an approximation parameter $\alpha \geq 1$, an \emph{$\alpha$-approximate decision proof labeling scheme ($\alpha$-ADPLS)} for a OptDGP $\Psi=\pif$ and a parameter $k\in \mathbb{Z}$ is defined in the same way as a DPLS for $\Psi$ and $k$ with the sole difference that the no-family is defined by setting 
\[
\mathcal{F}_{N}
\, = \,
\begin{cases}
\mathcal{U} \setminus \left\{
\gi\in \Pi\mid OPT_\Psi(G,\mathsf{I})\geq k/\alpha
\right\}
\, , &
\text{if $\Psi$ is a MinDGP} \\
\mathcal{U} \setminus \left\{
\gi\in \Pi\mid OPT_\Psi(G,\mathsf{I})\leq  \alpha\cdot k
\right\}
\, , &
\text{if $\Psi$ is a MaxDGP}
\end{cases} \, .
\]
We often refer to an $\alpha$-ADPLS without explicitly mentioning its
associated parameter $k$;
this should be interpreted with a universal quantifier over all parameters
$k \in \mathbb{Z}$.
\subsection{Related Work and Discussion}
\label{section:related-work}
Distributed verification is the task of locally verifying a global property of
a given configuration graph by means of a centralized prover and a distributed
verifier.
Various models for distributed verification have been introduced in the
literature including the PLS model \cite{pls} as defined in
Section~\ref{section:pls},
the \emph{locally checkable proofs (LCP)} model \cite{lcp},
and
the distributed complexity class \emph{non-deterministic local decision (NLD)}
\cite{nld1, nld2}.
Refer to \cite{survey} for a comprehensive survey on the topic of distributed
verification.

The current paper focuses on the PLS (and DPLS) model.
This model was introduced by Korman, Kutten, and Peleg in \cite{pls} and has
been extensively studied since then, see, e.g.,
\cite{mst, blin-2014, baruch-2016, error-sensitive-pls-Fraigniaud-2017,
	patt-shamir-perry-2017, redundancy-Feuilloley-2018,
	feuilloley2019introduction}.
A specific family of tasks that attracted a lot of attention in this regard is
that of designing PLSs for classic optimization problems.
Papers on this topic include \cite{mst}, where a PLS for minimum spanning tree
is shown to have a proof size of
$O (\log n \log W)$,
where $W$ is the maximum weight, and \cite{lcp}, where a PLS for maximum
weight matching in bipartite graphs is shown to have a proof size of
$O (\log W)$.

In parallel, numerous researchers focused on establishing impossibility
results for PLSs and DPLSs, usually derived from non-deterministic
communication complexity lower bounds
\cite{kushilevitz-communication-complexity}.
Such results are provided, e.g., in \cite{lower_bounds}, where a proof size of
$\tilde{\Omega}(n^2)$
is shown to be required for many classic optimization problems, and in
\cite{lcp}, where an
$\Omega(n^2/\log n)$
lower bound is established on the proof size of DPLSs for the problem of
deciding if the chromatic number is larger than $3$.
For the minimum spanning tree problem, the authors of \cite{mst} proved
that their
$O (\log n \log W)$
upper bound on the proof size is asymptotically optimal, relying on direct
combinatorial arguments.

The lower bounds on the proof size of PLSs (and DPLSs) for some optimization
problems have motivated the authors of \cite{aplspaper} to introduce the APLS
(and ADPLS) notion as a natural relaxation thereof.
This motivation is demonstrated by the task of verifying that the unweighted
diameter of a given graph is at most $k$:
As shown in \cite{aplspaper}, the diameter task admits a large gap between the
required proof size of a DPLS, shown to be
$\Omega(n / k)$,
and the proof sizes of
$(3 / 2)$-ADPLS
and $2$-ADPLS shown to be
$O(\sqrt{n} \log^2 n)$
and
$O(\log n)$,
respectively.
To the best of our knowledge, APLSs (and ADPLSs) have not been studied
otherwise until the current paper.

One of the generic methods developed in the current paper for the design of
APLSs for an abstract OptDGP relies on a \emph{primal dual} approach applied
to the linear program that encodes the OptDGP, after relaxing its integrality
constraints (see Section~\ref{section:primal-dual-method}).
This can be viewed as a generalization of a similar approach used in the
literature for concrete OptDGPs.
Specifically, this primal dual approach is employed in \cite{lcp} to obtain
their PLS for maximum weight matching in bipartite graphs with a proof size of
$O (\log W)$.
A similar technique is used by the authors of \cite{aplspaper} to achieve a
$2$-APLS for maximum weight matching in general graphs with the same proof
size.

While most of the PLS literature (including the current work) focuses on
deterministic schemes, an interesting angle that has been studied recently is
randomization in distributed proofs, i.e., allowing the verifier to reach its
decision in a randomized fashion.
The notion of \emph{randomized proof labeling schemes} was introduced in
\cite{rpls}, where the strength of randomization in the PLS model is
demonstrated by a universal scheme that enables one to reduce the amount of required communication
in a PLS exponentially by allowing a (probabilistic) one-sided error.
Another interesting generalization of PLSs is the \emph{distributed
	interactive proof} model, introduced recently in \cite{Kol-interactive-proofs}
and studied further in
\cite{parter-naor-yogev-2018-IP, Crescenzi-2019-IP, Fraigniaud-2019-IP}.

\paragraph*{On Sequential Efficiency.}

In this paper, we focus on sequentially efficient schemes, restricting the
prover and verifier to ``reasonable computations''.
We argue that beyond the interesting theoretical implications of this
restriction (see Section~\ref{section:contribution}), it also carries
practical justifications:
A natural application of PLSs is found in \emph{local checking} for
self-stabilizing algorithms \cite{Awerbuch-1991}, where the verifier's role is
played by the detection module and the prover is part of the correction module
\cite{pls}.
Any attempt to implement these modules in practice clearly requires sequential
efficiency on behalf of both the verifier and the prover (although, for the
latter, the sequential efficiency condition alone is not sufficient as the
correction module is also distributed).

While most of the PLSs presented in previous papers are naturally sequentially
efficient, there are a few exceptions.
One example of a scheme that may require intractable computations on the
verifier side is the universal PLS presented in \cite{pls} that enables the
verification of any decidable graph property with a label size of
$O (n^{2})$
simply by encoding the entire structure of the graph within the label.
A PLS that inherently relies on sequentially inefficient prover can be found,
e.g., in \cite{lcp}, where a scheme is constructed to decide if the graph
contains a Hamiltonian cycle.

\subsection{Our Contribution}
\label{section:contribution}
Our goal in this paper is to explore the power and limitations of APLSs and
ADPLS for OptDGPs.
We start by developing two generic methods:
a primal dual method for the design of sequentially efficient APLSs that
expands and generalizes techniques used by G{\"o}{\"o}s and Suomela \cite{lcp}
and Censor-Hillel, Paz, and Perry \cite{aplspaper};
and a method that exploits the local properties of centralized approximation
algorithms for the design of sequentially efficient ADPLSs.
Next, we establish black-box reductions between APLSs and ADPLSs for certain
families of OptDGPs.
Based (mainly) on these generic methods and reductions, we design a total of
twenty-two new sequentially efficient APLSs and ADPLSs for various classic
optimization problems;
refer to Tables \ref{apls-table} and \ref{adpls-table} for a summary of these
results.

On the negative side, we establish an
$\Omega (\log \kappa)$
lower bound on the proof size of a
$\frac{\kappa + 1}{\kappa}$-APLS
for maximum $b$-matching (in fact, this lower bound applies even for the
simpler case of maximum matching) and minimum edge cover in graphs of
odd-girth
$2 \kappa + 1$;
and an
$\Omega (\log n)$
lower bound on the proof size of a PLS for minimum edge cover in odd rings.
These lower bounds, that rely on combinatorial arguments and hold regardless
of sequential efficiency, match the proof size established in our
corresponding APLSs for these OptDGPs, thus proving their optimality.

Additional lower bounds are established under the restriction of the verifier
and/or prover to sequentially efficient computations, based on hardness
assumptions in (sequential) computational complexity theory.
Consider a OptDGP $\Psi$ that corresponds to an optimization problem that is
$\textrm{NP}$-hard to approximate within
$\alpha \geq 1$.
We first note that under the assumption that
$\textrm{NP} \neq \textrm{co-NP}$,
the yes-families of both an $\alpha$-APLS for $\Psi$ and an
$\alpha$-ADPLS for $\Psi$ (with some parameter $k\in\mathbb{Z}$)
are languages in the complexity class
$\textrm{co-NP} \setminus \textrm{NP}$.
Therefore, restricting the verifier to sequentially efficient computations
implies that $\Psi$ admits neither an $\alpha$-APLS, nor an $\alpha$-ADPLS,
with a polynomial proof size. This provides additional motivation for the study of APLSs and ADPLSs over
their exact counterparts.

Furthermore, the (weaker) assumption that
$\textrm{P} \neq \textrm{NP}$
suffices to rule out the existence of $\alpha$-ADPLS for $\Psi$ when
both the verifier and prover are required to be sequentially efficient.
This is due to the fact that the yes-family of an $\alpha$-ADPLS for $\Psi$ (with some parameter $k\in\mathbb{Z}$) is
a $\textrm{co-NP}$ complete language, combined with the trivial observation
that any sequentially efficient GPLS can be simulated by a centralized
algorithm in polynomial time.
We note that most of the OptDGPs considered in this paper correspond to
$\textrm{NP}$-hard optimization problems;
refer to Table~\ref{inapproximiability} for their known inapproximability
results with and without the unique games conjecture \cite{UGC-Khot}.

\subsection{Paper's Organization}
\label{section:organization}
The rest of the paper is organized as follows.
Following some preliminaries presented in Section~\ref{section:preliminaries},
our generic methods for the design of APLSs and ADPLSs are developed in
Section~\ref{section:methods}.
The reductions between APLSs and ADPLSs are presented in
Section~\ref{section:connections-between-ADPLSs-and-APLSs}.
Finally, the bounds we establish for concrete OptDGPS are established in
Section~\ref{section:results}.

\section{Preliminaries}
\label{section:preliminaries}

\paragraph*{Linear Programming and Duality.}
A \emph{linear program (LP)} consists of a linear objective function that one
wishes to optimize (i.e., minimize or maximize) subject to linear inequality
constraints.
The \emph{standard form} of a minimization (resp., maximization) LP is
$\min \{ \mathbf{c}^{\textrm{T}} \mathbf{x} \mid
\mathbf{A} \mathbf{x} \geq \mathbf{b} \,\land\, \mathbf{x} \geq \mathbf{0} \}$
(resp.,
$\max \{ \mathbf{c}^{\textrm{T}} \mathbf{x} \mid
\mathbf{A} \mathbf{x} \leq \mathbf{b} \,\land\, \mathbf{x} \geq \mathbf{0}
\}$),
where
$\mathbf{x} = \{ x_{j} \} \in \mathbb{R}^{\ell}$
is a vector of variables and
$\mathbf{A} = \{ a_{i, j} \} \in \mathbb{R}^{k \times \ell}$,
$\mathbf{b} = \{ b_{i} \} \in \mathbb{R}^{k}$,
and
$\mathbf{c} = \{ c_{j} \} \in \mathbb{R}^{\ell}$
are a matrix and vectors of coefficients.
An \emph{integer linear program (ILP)} is a LP augmented with integrality constraints. In Section \ref{section:results}, we formulate OptDGPs as LPs and ILPs. In the latter case, we often turn to a \emph{LP relaxation} of the problem, i.e., a LP obtained from an ILP by relaxing its integrality constraints.

Every LP admits a corresponding \emph{dual program} (in this context, we refer to the original LP as the \emph{primal program}). Specifically, for a minimization (resp., maximization) LP in standard form, its dual is a maximization (resp., minimization) LP, formulated as $\max \{ \mathbf{b}^{\textrm{T}} \mathbf{y} \mid \mathbf{A}^{\textrm{T}} \mathbf{y} \leq \mathbf{c} \,\land\, \mathbf{y} \geq \mathbf{0} \}$
(resp., $\min \{ \mathbf{b}^{\textrm{T}} \mathbf{y} \mid \mathbf{A}^{\textrm{T}} \mathbf{y} \geq \mathbf{c} \,\land\, \mathbf{y} \geq \mathbf{0} \}$).

LP duality has the following useful properties. Let $\mathbf{x}$ and $\mathbf{y}$ be feasible solutions to the primal and dual programs, respectively. The \emph{weak duality} theorem states that $\mathbf{c}^{\textrm{T}} \mathbf{x}\geq \mathbf{b}^{\textrm{T}} \mathbf{y}$ (resp., $\mathbf{c}^{\textrm{T}} \mathbf{x}\leq \mathbf{b}^{\textrm{T}} \mathbf{y}$). The \emph{strong duality} theorem states that $\mathbf{x}$ and $\mathbf{y}$ are optimal solutions to the primal and dual programs, respectively, if and only if $\mathbf{c}^{\textrm{T}} \mathbf{x}= \mathbf{b}^{\textrm{T}}\mathbf{y}$. The \emph{relaxed complementary slackness} conditions are stated as follows, for given parameters $\beta,\gamma\geq 1$. 
\begin{itemize} 
	\item[-] Primal relaxed complementary slackness:\\
	For every primal variable $x_{j}$, if $x_j>0$, then  $c_j/\beta\leq\sum\limits_{i=1}^{k} a_{i j}y_i \leq c_j$ (resp., $c_j\leq\sum\limits_{i=1}^{k} a_{i j}y_i \leq\beta\cdot c_j$).
	\item[-] Dual relaxed complementary slackness:\\
	For every dual variable $y_{i}$, if $y_i>0$, then $  b_i\leq\sum\limits_{j=1}^{\ell} a_{i j}x_j \leq \gamma \cdot b_i$ (resp., $b_i/\gamma\leq\sum\limits_{j=1}^{\ell} a_{i j}x_j \leq b_i$).
\end{itemize} 
If the (primal and dual) relaxed complementary slackness conditions hold, then it is guaranteed that $\mathbf{c}^{\textrm{T}} \mathbf{x}\leq\beta\cdot\gamma\cdot\mathbf{b}^{\textrm{T}} \mathbf{y}$ (resp., $\mathbf{c}^{\textrm{T}} \mathbf{x}\geq\frac{1}{\beta\cdot\gamma}\cdot\mathbf{b}^{\textrm{T}} \mathbf{y}$) which, combined with the aforementioned weak duality theorem, implies that $\mathbf{x}$ approximates an optimal primal solution by a multiplicative factor of $\beta\cdot\gamma$. Moreover, the relaxed complementary slackness conditions with parameters $\beta=\gamma=1$, often referred to simply as the \emph{complementary slackness} conditions, hold if and only if $\mathbf{x}$ and $\mathbf{y}$ are optimal.

Let $\Psi=\pif$ be a OptDGP that can be represented as an ILP. Let $P$ be its LP relaxation and $D$ the dual LP of $P$. Given parameters $\beta,\gamma\geq1$, we say that $\Psi$ is \emph{($\beta,\gamma$)-fitted} if for any optimal (integral) solution $\mathbf{x}$ for the ILP corresponding to $P$, there exists a feasible solution $\mathbf{y}$ for $D$ such that the relaxed primal and dual complementary slackness conditions hold for $\mathbf{x}$ and $\mathbf{y}$ with parameters $\beta$ and $\gamma$, respectively. 

\paragraph*{Comparison Schemes.}
Let $\mathcal{U}$ be the universe of IO graphs $\gio$ where
$\msfi:V\rightarrow\{0,1\}^*$ is an input assignment that encodes a unique id represented using $O(\log n)$ bits for each node $v\in V$ (possibly among other input components). 
For a function
$h : \{ 0, 1 \}^* \times \{ 0, 1 \}^* \rightarrow \mathbb{R}$
and parameter
$k \in \mathbb{Z}$,
an
\emph{$(h, k)$-comparison
	scheme} is a mechanism designed to decide if
$\sum_{v \in V}{h(\msfi(v), \msfo(v))} \geq k$
for a given IO graph
$\gio \in \mathcal{U}$.
Formally, an
$(h, k)$-comparison
scheme is defined as a GPLS over $\mathcal{U}$ by setting the yes-family to be
$\mathcal{F}_{Y} = \{ \gio \in \mathcal{U} \mid \sum_{v \in V}{h(\msfi(v),
	\msfo(v))} \geq k \}$
and the no family to be
$\mathcal{F}_{N} = \mathcal{U} \setminus \mathcal{F}_{Y}$.
Notice that the task of deciding if
$\sum_{v\in V} h(\msfi(v),\msfo(v)) \leq k$
can be achieved by a
$(-h, -k)$-comparison
scheme, where
$-h : \{ 0, 1 \}^* \times \{ 0, 1 \}^* \rightarrow \mathbb{R}$
is defined by setting
$-h(a, b) = -1 \cdot h(a, b)$
for every
$a, b \in \{ 0, 1 \}^*$.

The following lemma has been established by Korman et
al.~\cite[Lemma 4.4]{pls}.
\begin{lemma}\label{hk}
	Given a function $h : \{ 0, 1 \}^* \times \{ 0, 1 \}^* \rightarrow \mathbb{R}$ that is computable in polynomial time and an integer $k\in \mathbb{Z}$, there exists a sequentially efficient $(h, k)$-comparison scheme with proof size
	$O(\log n + H)$,
	where
	$H$ is the maximal number of bits required to represent $h(\msfi(v),\msfo(v))$ for any $v\in V$.
\end{lemma}

\paragraph*{Additional Definitions.}
A \emph{feasibility scheme} for a DGP $\Pi$ is a GPLS over the universe $\mathcal{U}=\{\gio\mid\gi\in \Pi\}$ with the yes-family $\mathcal{F}_{Y}=\Pi$ and the no-family $\mathcal{U}\setminus \mathcal{F}_{Y}$. The \emph{odd-girth} of a graph $G=(V,E)$ is the length of the shortest odd cycle contained in $G$.
\section{Methods}
\label{section:methods}
In this section, we present two generic methods that facilitate the design of
sequentially efficient APLSs and ADPLSs with small proof sizes for many OptDGPs. These methods are used in most of the results established later on in Section \ref{section:results}. 

\subsection{The Primal Dual Method}
\label{section:primal-dual-method}
LP duality theory can be a useful tool in the design of a $(\beta \cdot \gamma)$-APLS for a $(\beta,\gamma)$-fitted OptDGP $\Psi$ (as shown in \cite{aplspaper,lcp}). The main idea of this approach is to use the relaxed complementary slackness conditions to verify that the output assignment $\msfo:V\rightarrow\{0,1\}^*$ of a given IO graph $\gio$ is approximately optimal for $G$ and $\msfi$ with respect to $\Psi$. Specifically, the prover provides the verifier with a proof that there exists a feasible dual solution $\mathbf{y}$ within a multiplicative factor of $\beta \cdot \gamma$ from the primal solution $\mathbf{x}$ derived from the output assignment $\msfo$; the verifier then verifies the primal and dual feasibility of $\mathbf{x}$ and $\mathbf{y}$, respectively, as well as their relaxed complementary slackness conditions.

We take a particular interest in the following family of OptDGPs.
Consider a OptDGP
$\Psi = \pif$
that can be represented by an ILP that admits a LP relaxation $P$ whose matrix
form is given by the variable vector
$\mathbf{x} = \{ x_{j} \} \in \mathbb{R}^{\ell}$
and coefficient matrix and vectors
$\mathbf{A} = \{ a_{i, j} \} \in \mathbb{R}^{k \times \ell}$,
$\mathbf{b} = \{ b_{i} \} \in \mathbb{R}^{k}$,
and
$\mathbf{c} = \{ c_{j} \} \in \mathbb{R}^{\ell}$.
We say that $\Psi$ is \emph{locally verifiable} if for every IO graph
$\gio \in \Pi$,
there exist mappings
$v : [k] \rightarrow V$
and
$e : [\ell] \rightarrow E$
that satisfy the following conditions:
(1)
$a_{i, j} = 0$
for every
$i \in [k]$
and
$j \in [\ell]$
such that $e(j)$ is not incident on $v(i)$;
(2)
the variable $x_{j}$ is encoded in the local output $\msfo(u)$ of node
$u \in V$
for every
$j \in [k]$
such that $e(j)$ is incident on $u$;
and
(3)
the coefficients
$a_{i, j}$,
$a_{i', j}$,
$b_{i}$, and $c_{j}$ are either universal constants or encoded in the local
input $\msfi(u)$ of node
$u \in V$
for every
$i, i' \in [k]$
and
$j \in [\ell]$
such that
$v(i) = u$,
$v(i') = u'$,
and
$e(j) = (u, u')$.

The \emph{primal dual} method facilitates the design of an $\alpha$-APLS,
$\alpha = \beta \cdot \gamma$, for a $(\beta, \gamma)$-fitted and locally
verifiable OptDGP $\Psi=\pif$ whose goal is to determine for a given IO graph
$\gio$ if the output assignment $\msfo : V \rightarrow \{ 0, 1 \}^{*}$ is an
optimal (feasible) solution for the co-legal $G$ and $\msfi$ or $\alpha$-far
from being an optimal solution.
Let $\mathbf{x}$ be the primal variable vector encoded in the output
assignment $\msfo$.
If $\msfo$ is an optimal solution for $G$ and $\msfi$, then the prover uses a
sequential algorithm to generate a feasible dual variable vector $\mathbf{y}$
such that $\mathbf{x}$ and $\mathbf{y}$ meet the relaxed complementary
slackness conditions with parameters $\beta$ and $\gamma$ (such a dual
solution $\mathbf{y}$ exists as $\Psi$ is
$(\beta, \gamma)$-fitted).
The label assignment
$L : V \rightarrow \{ 0, 1 \}^*$
constructed by the prover assigns to each node
$u \in V$,
a label $L(u)$ that encodes the vector
$\mathbf{y}(u) = \langle y_i \mid v(i) = u \rangle$
of dual variables mapped to $u$ in the dual variable vector $\mathbf{y}$.

Consider some node
$u \in V$
of the given IO graph $\gio$.
The verifier at node $u$ extracts
(i)
the vector
$\mathbf{x}(u) = \langle x_{j} \mid u \in e(j) \rangle$
of primal variables mapped to edges incident on $u$ from the local output
$\msfo(u)$;
(ii)
the vector $\mathbf{y}(u)$ of dual variables mapped to $u$ from the label
$L(v)$;
(iii)
the vector
$\mathbf{y}^N(u) = \langle y_i \mid v(i) \in N(u) \rangle$
of dual variables mapped to $u$'s neighbors from the label vector $L^{N}(u)$;
and
(iv)
the vectors
$\mathbf{a}(u) = \langle a_{i, j} \mid  u \in e(j) \rangle$,
$\mathbf{b}(u) = \langle b_{i} \mid v(i) = u \rangle$,
and
$\mathbf{c}(u) = \langle c_{j} \mid u \in e(j) \rangle$
of coefficients mapped to $u$ and the edges incident on $u$ from the local
input $\msfi(u)$.

The verifier at node
$u \in V$
then proceeds as follows:
(1)
using
$\mathbf{x}(u)$,
$\mathbf{a}(u)$,
and
$\mathbf{b}(u)$,
the verifier verifies that the primal constraints that correspond to rows
$i \in [k]$
such that
$v(i) = u$
are satisfied;
(2)
using
$\mathbf{y}(u)$,
$\mathbf{y}^{N}(u)$,
$\mathbf{a}(u)$,
and
$\mathbf{c}(u)$,
the verifier verifies that the dual constraints that correspond to columns
$j \in [\ell]$
such that
$u \in e(j)$
are satisfied;
(3)
using
$\mathbf{x}(u)$,
$\mathbf{y}(u)$,
$\mathbf{y}^{N}(u)$,
$\mathbf{a}(u)$,
and
$\mathbf{c}(u)$,
the verifier verifies that the primal relaxed complementary slackness
conditions that correspond to primal variables $x_{j}$ such that
$u \in e(j)$
hold with parameter $\beta$;
and
(4)
using
$\mathbf{x}(u)$,
$\mathbf{y}(u)$,
$\mathbf{a}(u)$,
and
$\mathbf{b}(u)$,
the verifier verifies that the dual relaxed complementary slackness conditions
that correspond to dual variables $y_{i}$ such that
$v(i) = u$
hold with parameter $\gamma$.
If all four conditions are satisfied, then the verifier at node $u$ returns
$\true$;
otherwise, it returns $\false$.
Put together, the verifier accepts the IO graph $\gio$ if and only if
$\mathbf{x}$ and $\mathbf{y}$ are feasible primal and dual solutions that
satisfy the primal and dual relaxed complementary slackness conditions with
parameters $\beta$ and $\gamma$, respectively.

To establish the correctness of the $\alpha$-APLS, notice first that the
primal constraints are satisfied if and only if $\msfo$ is a feasible solution
for $G$ and $\msfi$.
Assuming that primal constraints are satisfied, if $\msfo$ is an optimal
solution for $G$ and $\msfi$, then the fact that $\Psi$ is
$(\beta, \gamma)$-fitted
implies that the verifier generates a feasible dual solution $\mathbf{y}$
such that the primal and dual relaxed complementary slackness conditions
are satisfied with parameters $\beta$ and $\gamma$.
Conversely, If $\mathbf{y}$ is a feasible dual solution and the primal and
dual relaxed complementary slackness conditions are satisfied with parameters
$\beta$ and $\gamma$, then $\mathbf{x}$ approximates the optimal primal
(fractional) solution within an approximation bound of
$\beta \cdot \gamma = \alpha$,
hence $\msfo$ approximates
$OPT_{\Psi}(G, \msfi)$
within the same approximation bound.

The proof size of a
$(\beta \cdot \gamma)$-APLS
for a
$(\beta, \gamma)$-fitted
and locally verifiable OptDGP $\Psi$, designed by means of the primal dual
method, is the maximum number of bits required to encode the vector
$\mathbf{y}(u)$
of dual variables mapped to a node
$u \in V$.
Let $r$ be the range of possible values assigned by the prover to a dual variable $y_{i}$.
We aim for schemes that minimize $r$. Particularly, for OptDGPs where the number of primal constraints mapped to
each node is bounded by a constant, this results in a
$(\beta \cdot \gamma)$-APLS
with a proof size of $O(\log r)$. 

In Section \ref{section:results} we present APLSs that are obtained using the primal
dual method.
We note that for all these APLSs, both the prover and verifier run in
polynomial sequential time, thus yielding sequentially efficient APLSs.
\subsection{The Verifiable Centralized Approximation Method}
\label{section:centralized-verifiable-approximation-method}
Consider some OptDGP
$\Psi = \pif$. We say that $\Psi$ is \emph{identified} if the input assignment $\msfi:V\rightarrow\{0,1\}^*$  encodes a unique id represented using $O(\log n)$ bits at each node $v\in V$ (possibly among other input components) for every IO graph $\gio$.

We say that 
$\Psi$
is \emph{decomposable} if there exists a function
$\lambda : \{ 0, 1 \}^* \times \{ 0, 1 \}^* \rightarrow \mathbb{R}$,
often referred to as a \emph{decomposition function}, such that 
$f(\gio) = \sum_{v \in V} \lambda(\msfi(v), \msfo(v))$
for every IO graph
$\gio \in \Pi$
(cf.\ the notion of semi-group functions in \cite{pls}).
Given an input and output assignments
$\msfi, \msfo : V \rightarrow \{ 0, 1 \}^*$,
let
$\lambda(\msfi, \msfo) = \sum_{v \in V} \lambda(\msfi(v), \msfo(v))$
denote the sum of the decomposition function values
$\lambda(\msfi(v), \msfo(v))$
over all nodes
$v \in V$.
Notice that the decomposition function is well defined for all bit string
pairs;
in particular, the definition of
$\lambda(\msfi, \msfo)$
does not require that the output assignment $\msfo$ is a feasible solution for
the graph $G$ and the input assignment $\msfi$.

Let $\Psi = \pif$ be a decomposable MinDGP (resp., MaxDGP)
with a decomposition function
$\lambda : \{ 0, 1 \}^* \times \{ 0, 1 \}^* \rightarrow \mathbb{R}$.
Given a legal input graph
$\gi\in \Pi$
and a parameter
$\alpha \geq 1$,
we say that a (not necessarily feasible) output assignment
$\mathsf{A} : V \rightarrow \{ 0, 1 \}^*$
is a \emph{decomposable $\alpha$-approximation} for $G$ and $\msfi$ if
$OPT_\Psi(G, \msfi)
\leq
\lambda(\msfi, \mathsf{A})
\leq
\alpha \cdot
OPT_\Psi(G, \msfi)$
(resp.,
$OPT_\Psi(G, \msfi) / \alpha
\leq
\lambda(\msfi, \mathsf{A})
\leq
OPT_\Psi(G, \msfi)$).

Fix some identified decomposable MinDGP (resp., MaxDGP)
$\Psi = \pif$
with a decomposition function
$\lambda : \{ 0, 1 \}^* \times \{ 0, 1 \}^* \rightarrow \mathbb{R}$.
The \emph{verifiable centralized approximation (VCA)} method facilitates the design
of an $\alpha$-ADPLS for $\Psi$ 
whose goal is to determine for a given legal input graph
$\gi \in \Pi$ and some parameter
$k \in \mathbb{Z}$
if every output assignment yields an objective value of at least (resp., at
most) $k$ or if there exists an output assignment that yields an objective
value smaller than
$k / \alpha$
(resp., larger than
$\alpha \cdot k$).
The ADPLSs designed by means of the VCA method are composed of two
verification tasks, namely, the \emph{approximation} task and the
\emph{comparison} task, so that the verifier accepts $\gi$ if and only if both
verification tasks accept.
The label
$L(v) = \langle L_\mathtt{approx}(v), L_\mathtt{comp}(v) \rangle$
assigned by the prover to each node
$v \in V$
is composed of the fields $L_\mathtt{approx}(v)$ and $L_\mathtt{comp}(v)$
serving the approximation task and the comparison task,
respectively.

In the approximation task, the prover runs a centralized algorithm
$\mathtt{ALG}$ that is guaranteed to produce a decomposable
$\alpha$-approximation
$\mathsf{A} : V \rightarrow \{ 0, 1 \}^*$
for graph $G$ and input assignment $\msfi$.
The $L_\mathtt{approx}(v)$ field of the label $L(v)$ assigned by the prover to
each node
$v \in V$
consists of both $\mathsf{A}(v)$ and a proof that the output assignment
$\mathsf{A}$ is indeed the outcome of the centralized algorithm $\mathtt{ALG}$.
The correctness requirement for this task is defined so that the verifier
accepts $\gi$ if and only if the field $L_\mathtt{approx}(v)$ encodes an
output assignment $\mathsf{A}$ that can be obtained using $\mathtt{ALG}$.

The purpose of the comparison task is to verify that
$\lambda(\msfi, \mathsf{A}) \geq k$
(resp.,
$\lambda(\msfi, \mathsf{A}) \leq k$),
where $\lambda$ is the decomposition function associated with the
(decomposable) MinDGP (resp., MaxDGP) $\Psi$ and $\mathsf{A}$ is the output
assignment encoded in the $L_\mathtt{approx}(v)$ fields of the labels $L(v)$
assigned to nodes
$v \in V$.
This is done by means of the
$(\lambda, k)$-comparison
scheme (resp., the
$(-\lambda, -k)$-comparison
scheme) presented in Section~\ref{section:preliminaries}.

The correctness of the $\alpha$-ADPLS for the MinDGP (resp., MaxDGP) $\Psi$
and the integer $k$ is established as follows.
If
$OPT_\Psi(G, \msfi) \geq k$
(resp.,
$OPT_\Psi(G, \msfi) \leq k$),
then the $L_\mathtt{approx}(v)$ field of the label $L(v)$ assigned by the
prover to each node
$v \in V$
encodes an output assignment
$\mathsf{A} : V \rightarrow \{ 0, 1 \}^*$
generated by the algorithm $\mathtt{ALG}$.
This means that $\mathsf{A}$ is a decomposable $\alpha$-approximation, thus
$\lambda(\msfi, \mathsf{A}) \geq OPT_\Psi(G, \msfi) \geq k$
(resp.,
$\lambda(\msfi, \mathsf{A}) \leq OPT_\Psi(G, \msfi) \leq k$)
and the verifier accepts $\gi$.
On the other hand, if
$OPT_\Psi(G, \msfi) < k / \alpha$
(resp.,
$OPT_\Psi(G, \msfi) > \alpha \cdot k$),
then for any decomposable $\alpha$-approximation $\mathsf{A}$, it holds that
$\lambda(\msfi, \mathsf{A}) \leq \alpha \cdot OPT_\Psi(G, \msfi) < k$
(resp.,
$\lambda(\msfi, \mathsf{A}) \geq OPT_\Psi(G, \msfi) / \alpha > k$),
hence the verifier rejects $\gi$ for any label assignment
$L : V \rightarrow \{ 0, 1 \}^*$.

The proof size of the $\alpha$-ADPLS designed via the VCA method is the
maximum size of a label
$L(v) = \langle L_\mathtt{approx}(v), L_\mathtt{comp}(v) \rangle$
assigned by the prover for a given input graph $\gi$ such that
$OPT_\Psi(G, \msfi) \geq k$
(resp.,
$OPT_\Psi(G, \msfi) \leq k$).
As discussed in Section~\ref{section:preliminaries}, it is guaranteed that the
$L_\mathtt{comp}(\cdot)$ fields are represented using
$O (\log n + H)$
bits, where
$H$ is an upper bound on the number of bits required to represent a $\lambda(\msfi(v),\msfo(v))$ value for any $v\in V$,
and
$\mathsf{A}$ is the decomposable $\alpha$-approximation generated by the
prover in the approximation task.
In Section~\ref{section:results}, we develop ADPLSs whose $L_\mathtt{approx}(\cdot)$
fields are also represented using
$|L_\mathtt{approx}(v)| = O (\log n + H)$
bits.
Moreover, the OptDGPs we consider admit some fixed parameter
$W \in \mathbb{Z}$
(typically an upper bound on the weights in the graph) such that
$H = O (\log n+\log W)$
which results in a proof size of
$O(\log n + \log W)$.

A desirable feature of the ADPLSs we develop in Section~\ref{section:results} is that
the centralized algorithms $\mathtt{ALG}$ employed in the approximation task
are efficient, hence the prover runs in polynomial time.
Since the (sequential) runtime of the verifier is also polynomial, it follows
that all our ADPLSs are sequentially efficient.

\section{Reductions Between APLSs and ADPLSs}
\label{section:connections-between-ADPLSs-and-APLSs}

\subsection{From an $\alpha$-ADPLS to an $\alpha$-APLS}
\label{section:from-adpls-to-apls}
Consider an identified decomposable MinDGP (resp., MaxDGP)
$\Psi = \pif$
with a decomposition function
$\lambda : \{ 0, 1 \}^* \times \{ 0, 1 \}^* \rightarrow \mathbb{R}$.
Let $p$ and $r$ be the proof sizes of a feasibility scheme for $\Pi$ and an
$\alpha$-ADPLS for $\Psi$,
respectively.
We establish the following lemma.

\begin{lemma}\label{lemma-from-adpls-to-apls}
	There exists an $\alpha$-APLS for $\Psi$ with a proof size of
	$O (p + r + \log n + H)$,
	where $H$ is the maximal number of bits required to represent
	$\lambda(\msfi(v), \msfo(v))$
	for any
	$v \in V$.
\end{lemma}
\begin{proof}
	Observe that if $\msfo$ is known to be a feasible solution for $G$ and
	$\msfi$, then the correctness requirements of an $\alpha$-APLS for the MinDGP
	(resp., MaxDGP) $\Psi$ are equivalent to those of an $\alpha$-ADPLS for $\Psi$
	and
	$k = f(\gio)$.
	That is, for a given IO graph
	$\gio\in \Pi$,
	if
	$OPT_{\Psi}(G, \msfi) \geq k$
	(resp.,
	$OPT_{\Psi}(G, \msfi) \leq k$),
	then $\msfo$ is an optimal solution for $G$ and $\msfi$ which requires the
	verifier of an $\alpha$-APLS to accept $\gio$;
	if
	$OPT_{\Psi}(G, \msfi) < k / \alpha$
	(resp.,
	$OPT_{\Psi}(G, \msfi) > \alpha \cdot k$),
	then $\msfo$ is at least $\alpha$-far from being optimal for $G$ and $\msfi$
	which requires the verifier of an $\alpha$-APLS to reject $\gio$.
	
	The design of an $\alpha$-APLS for $\Psi$ is thus enabled by taking the label
	assigned by the prover to each node
	$v \in V$
	to be
	$L(v) = \langle
	L_\mathtt{feas}(v), L_\mathtt{obj}(v), L_\mathtt{comp}(v), L_\mathtt{ADPLS}(v)
	\rangle$,
	where
	$L_\mathtt{feas}(v)$ is the $p$-bit label assigned to $v$ by the prover of the
	feasibility scheme for $\Pi$;
	$L_\mathtt{obj}(v) = f(\gio)$
	(note that all nodes are assigned with the same $L_\mathtt{obj}(\cdot)$
	field);
	$L_\mathtt{comp}(v)$ is the label constructed in the
	$(\lambda, L_\mathtt{obj}(v))$-comparison
	scheme (resp., the
	$(-\lambda, -L_\mathtt{obj}(v))$-comparison
	scheme) presented in Section~\ref{section:preliminaries};
	and
	$L_\mathtt{ADPLS}(v)$ is the $r$-bit label of an $\alpha$-ADPLS for $\Psi$ and
	$L_\mathtt{obj}(v)$.
	This label assignment allows the verifier to verify that
	(1)
	$\msfo$ is a feasible solution for $G$ and $\msfi$;
	(2)
	$f(\gio) \geq L_\mathtt{obj}(v)$
	(resp.,
	$f(\gio) \leq L_\mathtt{obj}(v)$)
	for each
	$v \in V$;
	and
	(3)
	the verifier of an $\alpha$-ADPLS for $\Psi$ and
	$k = f(\gio)$
	accepts the input graph $\gi$.
\end{proof}

Consider the OptDGPs presented in Section~\ref{section:results} in the context
of an $\alpha$-ADPLS with a proof size of $O(\log n+\log W)$.
We note that these OptDGPs admit sequentially efficient feasibility schemes
with a proof size of $O(\log n)$.
Specifically, for minimum weight vertex cover and minimum weight dominating
set a proof size of $1$ bit suffices;
for metric traveling salesperson, a feasibility scheme requires verifying that
a given solution is a Hamiltonian cycle which can be done efficiently with a
proof size of
$O(\log n)$
\cite{lcp};
and the feasibility scheme for minimum metric
Steiner tree requires verifying that a given solution is a tree that spans all
nodes of a given set which can be done efficiently with a proof size of
$O(\log n)$ \cite{pls}.
Since their objective functions are simply sums of weights, these OptDGPs also
admit natural decomposition functions whose images can be represented using
$O (\log n + \log W)$
bits assuming that $\msfo$ is a feasible output assignment.
Put together with Lemma \ref{lemma-from-adpls-to-apls}, we get that for each
sequentially efficient $\alpha$-ADPLS presented in
Section~\ref{section:results}, there exists a corresponding sequentially
efficient $\alpha$-APLS with a proof size of
$O(\log n + \log W)$.

\subsection{From an $\alpha$-APLS to an $\alpha$-ADPLS}
\label{section:adpls-locally-verifiable}
Consider an identified, locally verifiable, and $(\beta,\gamma)$-fitted OptDGP $\Psi$ with the mappings $v : [s] \rightarrow V$ and $e : [t] \rightarrow E$ that are associated with its LP relaxation $P$ whose matrix form is given by the variable vector
$\mathbf{x} = \{ x_{j} \} \in \mathbb{R}^{t}$
and coefficient matrix and vectors
$\mathbf{A} = \{ a_{i, j} \} \in \mathbb{R}^{s \times t}$,
$\mathbf{b} = \{ b_{i} \} \in \mathbb{R}^{s}$,
and
$\mathbf{c} = \{ c_{j} \} \in \mathbb{R}^{t}$. Define $D_u=\{i\mid v(i)=u\}$ for each $u\in V$ and let $d=\max_{u\in V}\{|D_u|\}$. Let $b_{\max}$ be the maximal number of bits required to represent $b_i$ for any $i\in[s]$. Let $\alpha=\beta\cdot\gamma$ and let $r$ be the proof size of an $\alpha$-APLS for $\Psi$ produced by the primal dual method. We obtain the following lemma.
\begin{lemma}\label{lemma-adpls-primal-dual}
	There exists an $\alpha$-ADPLS for $\Psi$ with a proof size of $O(\log n + \log d+b_{\max}+ r)$.
\end{lemma}
\begin{proof}
	We construct an $\alpha$-ADPLS for the MinDGP (resp., MaxDGP) $\Psi$ by means of the VCA method.
	Recall that an $\alpha$-APLS for $\Psi$ established by means of the primal dual method is defined so that the labels encode a feasible dual solution $\mathbf{y}\in \mathbb{R}^s$ that satisfies $OPT_\Psi(G,\msfi)\leq\alpha\cdot\mathbf{b}^\textrm{T}\mathbf{y}$ (resp., $OPT_\Psi(G,\msfi)\geq\frac{1}{\alpha}\cdot\mathbf{b}^\textrm{T}\mathbf{y}$). Define $\mathsf{A}(u)=\mathbf{y}(u)=\langle y_i \mid v(i)=u\rangle$ and $\lambda(\msfi(u),\mathsf{A}(u))=\alpha\cdot\sum_{i:v(i)=u}b_iy_i$ (resp., $\lambda(\msfi(u),\mathsf{A}(u))=\frac{1}{\alpha}\cdot\sum_{i:v(i)=u}b_iy_i$) for each $u\in V$. The prover sets the sub-label $L_\mathtt{approx}(u)=\mathsf{A}(u)=\mathbf{y}(u)$  associated with the approximation task for each node $u\in V$, which allows the verifier to verify that $\mathbf{y}$ is a feasible dual solution.
	
	\sloppy	
	For the correctness of this scheme, it suffices to show that $\mathsf{A}$ is a decomposable $\alpha$-approximation for $G$ and $\msfi$ (with respect to the decomposition function $\lambda$). 
	Note that $\mathbf{y}$ is defined so that it satisfies $OPT_\Psi(G,\msfi)\leq\alpha\cdot\mathbf{b}^\textrm{T}\mathbf{y}$ (resp., $OPT_\Psi(G,\msfi)\geq\frac{1}{\alpha}\cdot\mathbf{b}^\textrm{T}\mathbf{y}$); and weak duality implies that $\alpha\cdot\mathbf{b}^\textrm{T}\mathbf{y}\leq \alpha\cdot OPT_\Psi(G,\msfi)$ (resp., $\frac{1}{\alpha}\cdot\mathbf{b}^\textrm{T}\mathbf{y}\geq  \frac{1}{\alpha}\cdot OPT_\Psi(G,\msfi)$). It follows that $\mathsf{A}$ is a decomposable $\alpha$-approximation for $G$ and $\msfi$ since $\lambda(\msfi,\mathsf{A})=\sum_{u\in V}\lambda(\msfi(u),\mathsf{A}(u))=\alpha\cdot\sum_{u\in V}\sum_{i:v(i)=u}b_iy_i=\alpha\cdot\mathbf{b}^\textrm{T}\mathbf{y}$ (resp., $\lambda(\msfi,\mathsf{A})=\sum_{u\in V}\lambda(\msfi(u),\mathsf{A}(u))=\frac{1}{\alpha}\cdot\sum_{u\in V}\sum_{i:v(i)=u}b_iy_i=\frac{1}{\alpha}\cdot\mathbf{b}^\textrm{T}\mathbf{y}$). 
\end{proof}
\par\fussy

Observe that for the minimum edge cover problem presented in Section \ref{section:minimum-edge-cover} it holds that $b_{\max}=1$, $d=1$; and for the maximum $b$-matching problem presented in Section \ref{section:maximum-bmatching} it holds that $b_{\max}=O(\log W)$, $d=1$. These allow us to obtain the following results: (1) a $\frac{\kappa+1}{\kappa}$-ADPLS for minimum edge cover in graphs of odd-girth $2\kappa+1$ with a proof size of $O(\log n)$ based on Theorem \ref{EC-odd-girth}; (2) a DPLS for minimum edge cover in bipartite graphs with a proof size of $O(\log n)$ based on Theorem \ref{bipartite-ec}; (3) a $\frac{\kappa+1}{\kappa}$-ADPLS for maximum $b$-matching in graphs of odd-girth $2\kappa+1$ with a proof size of $O(\log n+\log W)$ based on Theorem \ref{bmatching-odd-girth}; and (4) a DPLS for maximum $b$-matching in bipartite graphs with a proof size of $O(\log n+\log W)$ based on Theorem \ref{bipartite-bmatching}.
\section{Bounds for Concrete OptDGPs}
\label{section:results}
\subsection{Minimum Edge Cover}
\label{section:minimum-edge-cover}
Given a graph $G=(V,E)$, an \emph{edge cover} is a subset $C\subseteq E$ of edges such that every node $v\in V$ is incident on at least one edge in $C$. A \emph{minimum edge cover} is an edge cover of minimal size. 

Given an edge cover $C$ in graph $G$, a node $v\in V$ is said to be \emph{tight} if it is incident on exactly one edge $e\in C$; otherwise it is said to be \emph{loose}. An \emph{interchanging} path is a simple path $P=\{e_{1}=(u,u'),\dots ,e_{\ell}=(v',v)\}$ between a loose node $u\in V$ and a node $v\in V$ that satisfies (1) $(u,u')\in C$; and (2) $|C\cap \{e_i,e_{i+1}\}|=1$ for all $i=1,2,\dots \ell-1$. We define $int(v)$ to be the length of a shortest interchanging path ending in $v$, defined to be $\infty$ if no such path exists, for each $v\in V$. In particular, $int(v)=0$ if and only if $v$ is loose.
\begin{lemma}\label{interchanging}
	Given an edge cover $C\subseteq E$ and a node $u\in V$, if $int(u)$ is odd, then for any node $v\in N(u)$, it holds that $int(v)\leq int(u)+1$.
\end{lemma}
\begin{proof}
	Let $P$ be an interchanging path of length $int(u)$ ending in $u$. Clearly, if a node $v\in N(u)$ precedes $u$ in $P$, then it follows that $int(v)< int(u)<int(u)+1$; otherwise $(u,v)\notin C$ (since $u$ is tight) which means that the path $P'=P\cup \{(u,v)\}$ is an interchanging path and thus $int(v)\leq int(u)+1$. 
\end{proof}

An \emph{inflating} path is an interchanging path $P=\{e_{1}=(u,u'),\dots ,e_{\ell}=(v',v)\}$ between two loose nodes $u,v\in V$, $u\neq v$, such that $e_{1},e_\ell\in C$.
\begin{lemma}\label{inflating}
	If $C$ is a minimum edge cover in a graph $G=(V,E)$, then there are no inflating paths in $G$.
\end{lemma}
\begin{proof}
	Assume towards a contradiction that $C$ is a minimum edge cover and there exists an inflating path $P=\{e_1=(u,u'),\dots e_\ell=(v',v)\}$ between two loose nodes $u,v\in V$. Let $P_{odd}=\{e_i\in P\mid i=1,3\dots \ell\}$ let $P_{even}=\{e_i\in P\mid i=2,4\dots \ell-1\}$ and let $\widetilde{C}=(C\setminus P_{odd})\cup P_{even}$. The edge set $\widetilde{C}$ is an edge cover (since $u$ and $v$ are loose in $C$) that satisfies $|\widetilde{C}|=C-1$ which contradicts $C$ being a minimum edge cover.
\end{proof}

\begin{theorem}\label{EC-odd-girth}
	For every $\kappa\in\mathbb{Z}_{>0}$, there exists a sequentially efficient $\frac{\kappa+1}{\kappa}$-APLS for minimum edge cover in graphs of odd-girth at least $2\kappa+1$ with a proof size of $\lceil \log (\kappa+1)\rceil$ bits.
\end{theorem}

\begin{proof}
	We provide a $\frac{\kappa+1}{\kappa}$-APLS by means of the primal dual
	method.
	Consider the following LP relaxation for the minimum edge cover problem
	\begin{equation}
	\begin{array}{rrclcl}\label{ec-primal}
	\displaystyle \min & \multicolumn{3}{l}{\underset{e\in E}{\sum} x_e}\\ 
	\textrm{s.t.} & \underset{e:v\in e}{\sum} x_e & \geq & 1, & &  v \in V\\
	& x_e & \geq & 0 ,& &  e \in E\\ 
	\end{array}
	\end{equation}
	and its dual LP
	\begin{equation}
	\begin{array}{rrclcl}\label{ec-dual}
	\displaystyle \max & \multicolumn{3}{l}{\underset{v\in V}{\sum}  y_v} \\
	\textrm{s.t.} & y_u+y_v & \leq & 1, & &  (u,v) \in E\\
	& y_v & \geq & 0 ,& &  v \in V. \\
	\end{array}
	\end{equation}
	Notice that the minimum edge cover problem is locally verifiable since each primal variable is mapped to an edge $e\in E$ and each primal constraint is mapped to a node $v\in V$.
	
	Given the proposed edge cover $C\subseteq E$ derived from the output assignment of the given IO graph $\gio$, let $U_{even}=\{v\in V\mid int(v)<\kappa\  \land\  int(v)\mod 2=0\}$ and let $U_{odd}=\{v\in V\mid int(v)<\kappa\  \land\  int(v)\mod 2=1\}$. The prover obtains a dual solution $\mathbf{y}\in \{0,1/(\kappa+1),\dots,\kappa/(\kappa+1)\}^n$ as follows: for each $v\in V$, set $y_{v}=\frac{1}{2(\kappa+1)}\cdot int(v)$ if $v\in U_{even}$; $y_{v}=1-\frac{1}{2(\kappa+1)}\cdot (int(v)+1)$ if $v\in U_{odd}$; and $y_{v}=\frac{\lceil{\kappa/2}\rceil}{\kappa+1}$ otherwise.
	
	Assuming that $C$ is a minimum edge cover we show that $\mathbf{y}$ is a feasible dual solution. Consider some edge $(u,v)\in E$. If $u,v\notin U_{odd}$, then $y_{u},y_{v}\leq 1/2$ and thus $(u,v)$ does not violate the feasibility of $\mathbf{y}$. We now consider the case that at least one endpoint of $(u,v)$ is in $U_{odd}$ and assume w.l.o.g.\  that $u\in U_{odd}$. Lemma \ref{inflating} states that there are no inflating paths in $G$, which combined with the fact that all odd cycles are larger than $2\cdot int(u)+1$ implies that $int(v)$ is even. Furthermore, from Lemma \ref{interchanging} we get that $int(v)\leq int(u)+1$. It follows that $y_{u}$ and $y_{v}$ satisfy their dual feasibility constraint since $y_{u}+y_{v}\leq 1-\frac{1}{2(\kappa+1)}\cdot (int(u)+1)+\frac{1}{2(\kappa+1)}\cdot (int(u)+1)= 1$ if $int(u)<\kappa-1$; and $y_{u}+y_{v}\leq 1-\frac{1}{2(\kappa+1)}\cdot \kappa+\frac{1}{2(\kappa+1)}\cdot \kappa= 1$ if $int(u)=\kappa-1$.
	
	Let $\mathbf{x}\in \{0,1\}^m$ be the primal solution that represents $C$. We show that if $C$ is a minimum edge cover, then the primal relaxed complementary slackness conditions hold with the parameter $\beta=(\kappa+1)/\kappa$. Consider some node $u\in V$. One of the following three cases apply for each edge $(u,v)\in E$ that satisfy $x_{u,v}=1$: (1) if $u\in U_{odd}$, then $int(v)=int(u)-1$; (2) if $u\in U_{even}$, then $int(v)=int(u)+1$; and (3) if $u\notin U_{odd} \cup U_{even}$, then $int(u)\geq \kappa$ and $int(v)\geq \kappa-1$. We observe that for each of these cases, $u$ and $v$ satisfy $y_{u}+y_{v}\geq \kappa/(\kappa+1)$, thus the primal relaxed complementary slackness conditions hold with the parameter $\beta=(\kappa+1)/\kappa$.
	
	As for the dual relaxed complementary slackness conditions, we note that $\mathbf{y}$ is defined so that if $y_{v}>0$, then $\sum_{e:v\in e}{x_{e}}=1$ for each node $v\in V$. Thus, it follows that the dual relaxed complementary slackness conditions are satisfied with the parameter $\gamma=1$.
	
	We note that the sequential runtime of both the prover and verifier is polynomial and that the labels assigned by the prover are taken from a range of $\kappa+1$ values, which, combined with the correctness of the primal dual method, completes our proof. 
\end{proof}

The following corollary follows directly from Theorem \ref{EC-odd-girth} by setting $\kappa=1$.
\begin{corollary}
	There exists a sequentially efficient $2$-APLS for minimum edge cover with a proof size of $1$ bit.
\end{corollary} 
We now show that the proof size of the $\frac{\kappa+1}{\kappa}$-APLS presented in Theorem \ref{EC-odd-girth} is optimal. We do so by constructing a corresponding lower bound. To that end, we consider the following problems.

Given a graph $G=(V,E)$, the \emph{leader election} problem is that of selecting a sole node $v\in V$ referred to as the \emph{leader}. We refer to a feasibility scheme for leader election simply as a \emph{leader election scheme}. Regarding leader election schemes, a result that was established by G{\"o}{\"o}s and Suomela \cite[Section 5.4]{lcp} is that the proof size of any leader election scheme in odd rings is $\Omega(\log n)$.

The \emph{two candidate leader election} problem handles the task of selecting
exactly one leader out of 2 candidate nodes $a,b\in V$. We refer to a
feasibility scheme for two candidate leader election simply as a \emph{two
	candidate leader election scheme}.
We obtain the following result regarding two candidate leader election scheme
in odd rings.
\begin{lemma}\label{2cle-result}
	There is a two candidate leader election scheme in odd rings with a proof size of $O(1)$.
\end{lemma}
\begin{proof}
	Let $\gio$ be an IO graph where $G=(V,E)$ is an odd ring with a sole leader $a\in V$ and a sole non-leader candidate $b\in V$. The prover assigns the label $L(a)=10$ to the leader $a$ and a label $L(v)\in\{0,1\}$ to every other node $v\in V\setminus\{a\}$, such that $L(v)\neq L(u)$ and $L(v)\neq L(u')$ for both of $v$'s neighbors $u,u'\in N(v)$. In other words, the label assignment $L:V\rightarrow\{0,1,10\}$ is a proper 3-coloring of the nodes such that the only node colored with the color $10$ is $a$. 
	
	The verifier at each node $v\in V$ with the neighbors $u,u'\in N(v)$, verifies that (1) $L(v)\in\{0,1,10\}$; (2) $L(v)=10$ if and only if $v$ is the leader; (3) if $L(v)=10$, then $L(u)\neq L(u')$; and (4) $L(v)\notin\{L(u),L(u')\}$. 
	
	For the correctness of this scheme we show that the verifier accepts $\gio$ if and only if there is exactly one leader out of the candidates $a$ and $b$. First we note that if no leader is elected, then there are two adjacent nodes $u,v\in V$ such that $L(u)=L(v)$. This follows from the fact that $G$ is not 2-colorable. In the case that both $a$ and $b$ are elected as leaders, there exists a candidate node $v\in\{a,b\}$ such that both of its neighbors are assigned the same label thus violating condition 3. Finally, if there is exactly one leader, then the label  assignment $L:V\rightarrow\{0,1,10\}$ constructed by the prover satisfies all the conditions checked by the verifier at each node $v\in V$.
\end{proof}
Consider an odd ring $G=(V,E)$. We observe that a minimum edge cover in $G$ admits exactly one loose node. This follows from the fact that the size of an edge cover $C\subseteq E$ with exactly one loose node is $\lceil n/2\rceil$ and every other edge cover is of size at least $\lceil n/2\rceil$. This leads to the following result.
\begin{lemma}\label{lemma-for-bound-ec}
For every $\kappa\in\mathbb{Z}_{>0}$, if there is a $\frac{\kappa+1}{\kappa}$-APLS for minimum edge cover in graphs of odd-girth at least $2\kappa+1$ with a proof size of $r$ bits, then there is a PLS for minimum edge cover in $(2\kappa+1)$-sized rings with a proof size of $r+O(1)$ bits. 
\end{lemma}
\begin{proof}
Let $L_\mathtt{APLS}:V\rightarrow\{0,1\}^r$ be the label assignment constructed by the prover of a $\frac{\kappa+1}{\kappa}$-APLS for minimum edge cover in graphs of odd-girth at least $2\kappa+1$. Given the proposed edge cover $C\subseteq E$ provided by means of the output assignment of a given IO graph $\gio$ with the underlying odd ring $G=(V,E)$, we construct a PLS as follows. Let $a\in V$ be the (unique) loose node in $G$ assuming that $C$ is optimal. The prover assigns the label $L(v)=\langle L_\mathtt{APLS}(v),L_\mathtt{2\_to\_1}(v)\rangle$, where for each node $v\in V$, the bit string $L_\mathtt{2\_to\_1}(v)$ is the constant size label of a two candidate leader election scheme in odd rings with $a$ as the leader (as presented Lemma \ref{2cle-result}). The verifier at each node $v\in V$, runs the local verification of both the APLS and the two candidate leader election scheme and returns $\true$ if and only if both schemes evaluate to $\true$.

Let $n=2\kappa+1$ be the number of nodes in $G$ and note that $n$ is the odd-girth of $G$. The size of a minimum edge cover is $\kappa+1$ and since $(\kappa+1)(\kappa+1)/\kappa<\kappa+3$ for all $\kappa>1$ (for completeness we state that it is not hard to obtain a PLS with constant proof size for the case $\kappa=1$), we get that if the verifier accepts $\gio$ with respect to the $\frac{\kappa+1}{\kappa}$-APLS, then $|C|\in \{\kappa+1,\kappa+2\}$. This means that there are at most 2 loose nodes. The correctness of this PLS is established by the fact that the two candidate leader election scheme guarantees that the verifier accepts $\gio$ if and only if there is exactly one loose node in $G$.
\end{proof}
\begin{corollary}
There exists a PLS for minimum edge cover in odd rings with a proof size of $O(\log n)$.
\end{corollary}
\begin{theorem}\label{bound_for_pls_ec}
The proof size of a PLS for minimum edge cover in $(2\kappa+1)$-sized rings is  $\Omega(\log n)$.
\end{theorem}
\begin{proof}
Assuming that there exists a PLS for minimum edge cover in $(2\kappa+1)$-sized rings with a proof size of $o(\log n)$ bits implies a leader election scheme in odd rings with a proof size of $o(\log n)$ simply by having the prover encode a minimum edge cover with the elected leader as the only loose node. This contradicts the fact that the proof size of any leader election scheme in odd rings is $\Omega(\log n)$.
\end{proof}	

\begin{theorem}\label{bound_for_apls_ec}
	For every $\kappa\in\mathbb{Z}_{>0}$, the proof size of a $\frac{\kappa+1}{\kappa}$-APLS for minimum edge cover in graphs of odd-girth at least $2\kappa+1$ is $\Omega(\log \kappa)$.
\end{theorem}
\begin{proof}
	Assume towards a contradiction that there exists a $\frac{\kappa+1}{\kappa}$-APLS for minimum edge cover in graphs of odd-girth at least $2\kappa+1$ with a proof size of $o(\log \kappa)$ bits. As Lemma \ref{lemma-for-bound-ec} suggests, this implies that there is a PLS for minimum edge cover in odd rings with a proof size of $o(\log n)$ which contradicts Theorem \ref{bound_for_pls_ec}.
\end{proof}

\begin{theorem}\label{ec-ring-dpls}
	There exists a sequentially efficient DPLS for minimum edge cover in odd rings with a proof size of $O(\log n)$.
\end{theorem}

\begin{proof}
	Since the size of a minimum edge cover in odd rings is $\lceil n/2\rceil$, it suffices to verify that $\lceil n/2\rceil\geq k$. This can be done by means of the $(h,k)$-comparison scheme presented in Section \ref{section:preliminaries} (note that this scheme is required as we assume that the nodes do not know the exact value of $n$).
\end{proof}

\begin{theorem}\label{bipartite-ec}
	There exists a sequentially efficient PLS for minimum edge cover in bipartite graphs with a proof size of $1$ bit.
\end{theorem}

\begin{proof}
	Consider the primal and dual LPs formulated in (\ref{ec-primal}) and (\ref{ec-dual}), respectively. It is well known that for the instance of bipartite graphs $G=(V,E)$ the primal and dual LPs admit optimal integral solutions $\mathbf{x}\in\{0,1\}^m$ and $\mathbf{y}\in \{0,1\}^n$, respectively. Moreover, the optimal integral dual solution $\mathbf{y}$ can be obtained by means of a sequential polynomial time algorithm and encoded by means of the 1 bit label $L(v)=y_{v}$ for each $v\in V$. To complete the design of this PLS we simply use the construction of a $1$-APLS for this locally verifiable $(1,1)$-fitted problem by means of the primal dual method.  
\end{proof}

\subsection{Maximum $b$-Matching}
\label{section:maximum-bmatching}
Consider a graph $G=(V,E)$ associated with a function $b:V\rightarrow
\{1,\dots,W\}$.
A \emph{$b$-matching} is a mapping $\mu:E\rightarrow\{0,1,\dots,W\}$ that satisfies $\sum_{e:v\in e}\mu(e)\leq b(v)$ for each $v\in V$. A \emph{maximum $b$-matching} is a $b$-matching $\mu$ that maximizes $\sum_{e\in E}\mu(e)$.

Given a $b$-matching $\mu$ in graph $G$, a node $v\in V$ is said to be \emph{matched} if  $\sum_{e:v\in e}\mu(e)= b(v)$; otherwise it is said to be \emph{available}. An \emph{alternating} path is a simple path $P=\{e_{1}=(u,u'),\dots ,e_{\ell}=(v',v)\}$ between an available node $u\in V$ and a node $v\in V$ that satisfies $\mu(e_i)>0$ for all even $i\in [\ell]$. We define $alt(v)$ to be the length of a shortest alternating path ending in $v$, defined to be $\infty$ if no such path exists, for each $v\in V$. In particular, $alt(v)=0$ if and only if $v$ is available.
\begin{lemma}\label{alternating}
	Given a $b$-matching $\mu$ and a node $u\in V$, if $alt(u)$ is even, then for any node $v\in N(u)$, it holds that $alt(v)\leq alt(u)+1$.
\end{lemma}
\begin{proof}
	Let $P$ be an alternating path of length $alt(u)$ ending in $u$. Clearly, if a node $v\in N(u)$ precedes $u$ in $P$, then it follows that $alt(v)< alt(u)<alt(u)+1$; otherwise the path $P'=P\cup \{(u,v)\}$ is an alternating path and thus $alt(v)\leq alt(u)+1$.  
\end{proof}
An \emph{augmenting} path is an alternating path $P=\{e_{1}=(u,u'),\dots ,e_{\ell}=(v',v)\}$ of odd length between two available nodes $u,v\in V$.
\begin{lemma}\label{augmenting}
	If $\mu$ is a maximum $b$-matching in a graph $G=(V,E)$, then there are no augmenting paths in $G$.
\end{lemma}
\begin{proof}
	Assume towards a contradiction that $\mu$ is a maximum $b$-matching and there exists an augmenting path $P=\{e_1=(u,u'),\dots e_\ell=(v',v)\}$ between two available nodes $u,v\in V$. Let $P_{odd}=\{e_i\in P\mid i=1,3\dots \ell\}$ and let $P_{even}=\{e_i\in P\mid i=2,4\dots \ell-1\}$. Consider the mapping $\mu'$ obtained by setting $\mu'(e)=\mu(e)$ for each $e\in E\setminus P$; $\mu'(e)=\mu(e)+1$ for each $e\in P_{odd}$; and $\mu'(e)=\mu(e)-1$ for each $e\in P_{even}$. Since $u$ ad $v$ are available in $\mu$, we get that $\mu'$ is a (feasible) $b$-matching that satisfies $\sum_{e\in E}\mu'(e)=\sum_{e\in E}\mu(e)+1$ which contradicts $\mu$ being a maximum $b$-matching.
\end{proof}
\begin{theorem}\label{bmatching-odd-girth}
	For every $\kappa\in\mathbb{Z}_{>0}$, there exists a sequentially efficient $\frac{\kappa+1}{\kappa}$-APLS for $b$-matching in graphs of odd-girth at least $2\kappa+1$ with a proof size of $\lceil \log (\kappa+1)\rceil$ bits.
\end{theorem}
\begin{proof}
	We provide a $\frac{\kappa+1}{\kappa}$-APLS by means of the primal dual method. Consider the following LP relaxation for the maximum $b$-matching problem
	\begin{equation}
	\begin{array}{rrclcl}\label{bmatching-primal}
	\displaystyle \max & \multicolumn{3}{l}{\underset{e\in E}{\sum} x_e} \\
	\textrm{s.t.} & \underset{e:v\in e}{\sum} x_e & \leq & b(v), & &  v \in V\\
	& x_e & \geq & 0 ,& &  e \in E \\
	\end{array}
	\end{equation}
	and its dual LP
	\begin{equation}
	\begin{array}{rrclcl}\label{bmatching-dual}
	\displaystyle \min & \multicolumn{3}{l}{\underset{v\in V}{\sum} b(v)\cdot y_v} \\
	\textrm{s.t.} & y_u+y_v & \geq & 1, & &  (u,v) \in E\\
	& y_v & \geq & 0 ,& &  v \in V .\\
	\end{array}
	\end{equation}
	Notice that the minimum edge cover problem is locally verifiable since each primal variable is mapped to an edge $e\in E$ and each primal constraint is mapped to a node $v\in V$.
	
	Given the proposed $b$-matching $\mu:E\rightarrow\{0,1,\dots,W\}$ derived from the output assignment of the given IO graph $\gio$, let $U_{even}=\{v\in V\mid alt(v)<\kappa\  \land\  alt(v)\mod 2=0\}$ and let $U_{odd}=\{v\in V\mid alt(v)<\kappa\  \land\  alt(v)\mod 2=1\}$. The prover obtains a dual solution $\mathbf{y}\in \{0,1/\kappa,\dots,(\kappa-1)/\kappa,1\}^n$ as follows: for each $v\in V$, set $y_{v}= alt(v)/2\kappa$ if $v\in U_{even}$; $y_{v}=1-(alt(v)-1)/2\kappa$ if $v\in U_{odd}$; and $y_{v}=(\lceil{\kappa/2}\rceil)/\kappa$ otherwise.
	
	Assuming that $\mu$ is a maximum $b$-matching we show that $\mathbf{y}$ is a feasible dual solution. Consider some edge $(u,v)\in E$. If $u,v\notin U_{even}$, then $y_{u},y_{v}\geq 1/2$ and thus $(u,v)$ does not violate the feasibility of $\mathbf{y}$. Consider the case that at least one endpoint of $(u,v)$ is in $U_{even}$ and assume w.l.o.g.\  that $u\in U_{even}$. Lemma \ref{augmenting} states that there are no augmenting paths in $G$, which combined with the fact that all odd cycles are larger than $2\cdot alt(u)+1$ implies that $alt(v)$ is odd. Furthermore, from Lemma \ref{alternating} we get that $alt(v)\leq alt(u)+1$. It follows that $y_{u}$ and $y_{v}$ satisfy their dual feasibility constraint since $y_{u}+y_{v}\geq alt(u)/2\kappa+1-alt(u)/2\kappa=1$ if $alt(u)<\kappa-1$; and $y_{u}+y_{v}\geq (\kappa-1)/2\kappa+(\kappa+1)/2\kappa=1$ if $alt(u)=\kappa-1$. 
	
	Let $\mathbf{x}\in \{0,1\}^m$ be the primal solution that represents $\mu$. We show that if $\mu$ is a maximum $b$-matching, then the primal relaxed complementary slackness conditions hold with the parameter $\beta=(\kappa+1)/\kappa$. Consider some node $u\in V$. One of the following three cases apply for each edge $(u,v)\in E$ that satisfy $x_{u,v}>0$: (1) if $u\in U_{odd}$, then $alt(v)$ is even and $alt(v)\leq alt(u)+1$; (2) if $u\in U_{even}$, then $alt(v)\in \{alt(u)-1,alt(u)+1\}$; and (3) if $u\notin U_{odd} \cup U_{even}$, then $alt(u)\geq \kappa$ and $alt(v)\geq \kappa-1$. We observe that for each of these cases, $u$ and $v$ satisfy $y_{u}+y_{v}\leq (\kappa+1)/\kappa$, thus the primal relaxed complementary slackness conditions are met with the parameter $\beta=(\kappa+1)/\kappa$.
	
	Regarding the dual relaxed complementary slackness conditions, $\mathbf{y}$ is defined so that if $y_{v}>0$, then $\sum_{e:v\in e}{x_{e}}=b(v)$ for each node $v\in V$. This implies that the dual relaxed complementary slackness conditions are satisfied with the parameter $\gamma=1$.
	
	We note that the sequential runtime of both the prover and verifier is polynomial and that the labels assigned by the prover are taken from a range of $\kappa+1$ values, which, combined with the correctness of the primal dual method, completes our proof. 
\end{proof}

The following corollary follows directly from Theorem \ref{bmatching-odd-girth} by setting $\kappa=1$.
\begin{corollary}
	There exists a sequentially efficient $2$-APLS for maximum $b$-matching with a proof size of $1$ bit.
\end{corollary} 

We go on to show that this proof size is optimal by obtaining a lower bound that applies even for the case of \emph{maximum matching}, i.e., maximum $b$-matching with $b(v)=1$ for each $v\in V$. To that end, we use the structure of an odd ring graph. As established by G{\"o}{\"o}s and Suomela \cite[Section 5.4]{lcp}, a PLS for maximum matching in odd rings requires a proof size of $\Omega (\log n)$. This leads us to the following result.
\begin{theorem}\label{bmatching-bound}
	The proof size of a $\frac{\kappa+1}{\kappa}$-APLS for maximum matching in graphs of odd-girth at least $2\kappa+1$ is $\Omega (\log \kappa)$.
\end{theorem}
\begin{proof}
	Consider a $(2\kappa+1)$-sized ring $G=(V,E)$ and note that the size of a maximum matching in $G$ is $\kappa$. It follows that applying a $\frac{\kappa+1}{\kappa}$-APLS  for maximum matching in a $(2\kappa+1)$-sized ring results in a PLS for maximum matching since $\frac{\kappa}{\kappa+1}\cdot \kappa>\kappa-1$. This means that a $\frac{\kappa+1}{\kappa}$-APLS for maximum matching in graphs of odd-girth at least $2\kappa+1$ with a proof size of $o(\log \kappa)$ is not possible as it would imply a PLS for maximum matching in odd rings with a proof size of $o(\log n)$. 
\end{proof}
\begin{theorem}\label{bipartite-bmatching}
	There exists a sequentially efficient PLS for maximum $b$-matching in bipartite graphs with a proof size of $1$ bit.
\end{theorem}
\begin{proof}
	Consider the primal and dual LPs formulated in (\ref{bmatching-primal}) and (\ref{bmatching-dual}), respectively. For the instance of bipartite graphs $G=(V,E)$, the primal and dual LPs admit optimal integral solutions $\mathbf{x}\in\{0,1\}^m$ and $\mathbf{y}\in \{0,1\}^n$, respectively. Moreover, the optimal integral dual solution $\mathbf{y}$ can be obtained by means of a sequential polynomial time algorithm and encoded by means of the 1 bit label $L(v)=y_{v}$ for each $v\in V$. To complete the design of this PLS we simply use the construction of a $1$-APLS for this locally verifiable $(1,1)$-fitted problem by means of the primal dual method.  
\end{proof}

\subsection{Minimum Weight Vertex Cover}
\label{section:min-weight-vc}
Consider a graph $G=(V,E)$ associated with a node-weight function $w:V\rightarrow\{1,\dots W\}$. A \emph{vertex cover} is a subset $U\subseteq V$ of nodes such that every edge $e\in E$ has at least one endpoint in $U$. A \emph{minimum weight vertex cover} is a vertex cover $U$ that minimizes $w(U)=\sum_{u\in U}w(u)$. 
\begin{theorem}\label{vc-adpls}
	There exists a sequentially efficient $2$-ADPLS for minimum weight vertex cover with a proof size of $O(\log n+\log W)$.
\end{theorem}
\sloppy
\begin{proof}
	We provide a 2-ADPLS by means of the VCA method (we assume that the nodes are assigned with unique ids represented using $O(\log n)$ bits). For the approximation task, consider the following 2-approximation algorithm (this algorithm resembles the one presented in \cite[Chapter 2]{vazirani}):
	\begin{enumerate}
		\item 
		Initialization: $U\leftarrow\emptyset$; $w'(v)\leftarrow w(v), c(v)\leftarrow\bot,\ \forall v\in V$.
		\item While $U$ is not a vertex cover do:
		\begin{enumerate}
			\item arbitrarily choose an edge $(u,v)\in E$ such that $u,v\notin U$.
			\item $a\leftarrow \arg \min \{w'(u),w'(v)\}$
			\item $b\leftarrow \{u,v\}\setminus\{a\}$
			\item $c(a)\leftarrow b$
			\item $U\leftarrow U\cup \{a\}$
			\item $w'(b)\leftarrow w'(b)-w'(a)$
		\end{enumerate}
	\end{enumerate}
	Following that, the prover constructs the sub-label 
	$L_\mathtt{approx}(v)=\langle L_\mathtt{approx}^1(v), L_\mathtt{approx}^2(v),L_\mathtt{approx}^3(v),L_\mathtt{approx}^4(v)\rangle$ for each node $v\in V$, where $L_\mathtt{approx}^1(v)$ is a bit indicating if $v\in U$; $L_\mathtt{approx}^2(v)=id(v)$; $L_\mathtt{approx}^3(v)=id(c(v))$ (we assume that
	$id(\bot) = \bot$); and $L_\mathtt{approx}^4(v)=w'(v)$.
	\par\fussy
	
	To complete the approximation task, the verifier needs to verify that $L_\mathtt{approx}^1(\cdot)$ encodes a decomposable 2-approximation obtained by the algorithm above. Let $C_v=\{u\in N(v)\mid L_\mathtt{approx}^3(u)=id(v)\}$. The verifier at each node $v\in V$, simply verifies that: (1) if $L_\mathtt{approx}^1(v)=0$, then $L_\mathtt{approx}^1(u)=1$ for all $u\in N(v)$; and (2) $w(v)=L_\mathtt{approx}^4(v)+\sum_{u\in C_v}L_\mathtt{approx}^4(u)$. Correctness follows from the correctness of the algorithm above as well as the VCA method. Note that the sequential runtimes of both the prover and verifier are polynomial which completes our proof.
\end{proof}

\subsection{Minimum Weight Dominating Set}
\label{section:min-weight-ds}
Consider a graph $G=(V,E)$ associated with a node-weight function $w:V\rightarrow\{1,\dots W\}$. A \emph{dominating set} is a subset $U\subseteq V$ of nodes such that $U\cap (\{v\}\cup N(v))\neq \emptyset$ for each node $v\in V$. A \emph{minimum weight dominating set} is a dominating set that minimizes $w(U)=\sum_{v\in U}w(v)$.
\begin{theorem}\label{ds-adpls}
	There exists a sequentially efficient $(\ln (n)+1)$-ADPLS for minimum weight dominating set with a proof size of $O(\log n+\log W)$.
\end{theorem}
\begin{proof}
	We obtain a $(\ln (n)+1)$-ADPLS by means of the VCA method. For the approximation task, the prover obtains a $(\ln (n)+1)$-approximation for minimum weight dominating set by means of the following greedy algorithm:
	\begin{enumerate}
		\item 
		Initialization: $U\leftarrow\emptyset$; $B\leftarrow\emptyset$.
		\item While $U$ is not a dominating set do:
		\begin{enumerate}
			\item for each node $v\notin U$, let $C_v=(N(v)\cup\{v\})\setminus B$.
			\item find a node $z\notin U$ that minimizes $w(z)/w(C_z)$.
			\item set $d(u)=w(z)/w(C_z)$ for each node $u\in C_z$
			\item $U\leftarrow U\cup\{z\}$
			\item $B\leftarrow B\cup C_z$
		\end{enumerate}
	\end{enumerate}
	Following that, the prover sets $L_\mathtt{approx}(v)=d(v)$ for each $v\in V$. Note that $d$ satisfies $\sum_{v\in V}d(v)=w(U)$ which means that it can be used in a decomposable $(\ln(n)+1)$-approximation.
	
	The verifier uses dual fitting in order to verify that the $L_\mathtt{approx}(\cdot)$ field encodes a decomposable $(\ln(n)+1)$-approximation. Consider the dual LP that corresponds to the LP relaxation of minimum weight dominating set
	\begin{equation}
	\begin{array}{rrclcl}\label{mds-dual}
	\displaystyle \max & \multicolumn{3}{l}{\sum\limits_{v\in V}y_v}\\ 
	\text{s.t.}& y_v + \sum\limits_{u\in N(v)}y_u & \leq &w(v) ,& &  v \in V\\
	& y_v & \geq & 0 ,& &  v\in V\\ 
	\end{array}
	\end{equation}
	and let $\mathbf{y}$ be the dual solution obtained by setting $y_{v}=L_\mathtt{approx}(v)/(\ln (n)+1)$ for each $v\in V$. It holds that if $L_\mathtt{approx}(v)$ actually encodes a $(\ln(n)+1)$-approximation obtained by the algorithm above, then $\mathbf{y}$ is a feasible solution. The verifier completes the approximation task by verifying the feasibility of $\mathbf{y}$. To complete our proof, note that the sequential runtime of both the prover and verifier is polynomial. 
\end{proof}

\subsection{Metric Traveling Salesperson}
\label{section:tsp}
Consider a metric graph $G=(V,E)$ associated with an edge-weight function $w:E\rightarrow\{1,\dots W\}$. The \emph{metric traveling salesperson problem (metric TSP)} is to find a Hamiltonian cycle $C$ in $G$ that minimizes $w(C)=\sum_{e\in C}w(e)$. 
\begin{theorem}\label{tsp-adpls}
	There exists a sequentially efficient $2$-ADPLS for metric TSP with a proof size of $O(\log n+\log W)$.
\end{theorem}
\begin{proof}
	We use the VCA method. Let $T$ and $C_\mathtt{OPT}$ be a minimum spanning tree and an optimal TSP tour in $G$, respectively. It follows that $w(C_\mathtt{OPT})\leq 2 \cdot w(T)\leq 2\cdot w(C_\mathtt{OPT})$ (\cite[Chapter 3]{vazirani}). This means that we can use $T$ in a decomposable $2$-approximation. Verification that $T$ is in fact a minimum spanning tree in a metric graph $G$ can be done with proof size of $O(\log n +\log W)$ using a scheme presented by Patt-Shamir and Perry in \cite[Theorem 6]{patt-shamir-perry-2017}. 
\end{proof}

\subsection{Minimum Metric Steiner Tree}
\label{section:steiner}
Consider a metric graph $G=(V,E)$ associated with an edge-weight function $w:E\rightarrow\{1,\dots W\}$ and let $S\subseteq V$ be a subset $S\subseteq V$ of \emph{terminal} nodes. The \emph{minimum metric Steiner tree} problem is to find a tree $T=(V_T,E_T)$ that minimizes $w(T)=\sum_{e\in E_T}w(e)$ such that $S\subseteq V_T\subseteq V$ and $E_T\subseteq E$.
\begin{theorem}\label{steiner-adpls}
	There exists a sequentially efficient $2$-ADPLS for minimum metric Steiner tree with a proof size of $O(\log n+\log W)$.
\end{theorem}
\begin{proof}
	We use the VCA method. Let $T_\mathtt{OPT}$ be an optimal Steiner tree. Consider the subset $E_S=\{(u,v)\in E\mid u,v\in S\}$ of edges with both endpoints in $S$ and the sub-graph $G_S=(S,E_S)$. Let $T_S$ be a minimum spanning tree in $G_S$. It follows that $w(T_\mathtt{OPT})\leq w(T_S)\leq 2\cdot w(T_\mathtt{OPT})$ \cite[Chapter 3]{vazirani}. This means that we can use $T_S$ in a decomposable $2$-approximation. To verify that $T_S$ is in fact a minimum spanning tree of $G_S$ we  use the scheme presented in \cite[Theorem 6]{patt-shamir-perry-2017}. 
\end{proof}

\subsection{Maximum Flow}
\label{section:flow}
Consider a directed graph $G=(V,E)$ with a source $s\in V$, a sink $t\in V$, and edge capacities $c:E\rightarrow\{0,1,\dots,W\}$. A \emph{flow} is a function $f:E\rightarrow\{0,1,\dots ,W\}$ satisfying: (1) $f(u,v)\leq c(u,v)$ for each edge $(u,v)\in E$; and (2) $\sum_{u:(u,v)\in E}f(u,v)=\sum_{z:(v,z)\in E}f(v,z)$ for each node $v\in V\setminus\{s,t\}$. A \emph{maximum flow} is a flow that maximizes $\sum_{v:(v,t)\in E}f(v,t)$.
\begin{theorem}\label{flow-pls}
	There exists a sequentially efficient PLS for maximum flow with a proof size of $1$ bit.
\end{theorem}

\begin{proof} 
	Consider the following LP relaxation obtained by adding an artificial edge $(t,s)$ with infinite capacity and setting the objective to maximize $x_{t,s}$ (cf.\ Chapter 12 of \cite{vazirani})
	\begin{equation}
	\begin{array}{rrclcl}\label{flow-primal}
	\displaystyle \max & \multicolumn{3}{l}{x_{t,s}}\\ 
	\text{s.t.}& x_{u,v} & \leq &c_{u,v} ,& &  (u,v) \in E\\
	&\underset{u:(u,v)\in E}{\sum} x_{u,v}-\underset{z:(v,z)\in E}{\sum}x_{v,z} & \leq & 0, & &  v \in V\\
	& x_{u,v} & \geq & 0 ,& &  (u,v)\in E\\ 
	\end{array}
	\end{equation}
	and its dual LP
	\begin{equation}
	\begin{array}{rrclcl}\label{flow-dual}
	\displaystyle \min & \multicolumn{3}{l}{\underset{(u,v)\in E}{\sum}  c_{u,v}\cdot y_{u,v}} \\
	\textrm{s.t.} & y_{u,v}-y_u+y_v & \geq & 0, & &  (u,v) \in E\\
	& y_{s}-y_{t}&\geq& 1\\
	& y_{u,v} & \geq & 0 ,& &  (u,v) \in E\\
	& y_v & \geq & 0 ,& &  v \in V .\\
	\end{array}
	\end{equation}
	
	Let $\mathbf{x}$ be the primal solution derived from the output assignment of the given IO graph $\gio$. The max-flow min-cut theorem states that if $\mathbf{x}$ represents a maximum flow and $\mathbf{y}$ is a dual solution that represents a minimum $(s,t)$-cut, then they are optimal solutions for the primal and dual LPs, respectively. To that end, the prover uses a sequential algorithm to find a minimum $(s,t)$-cut $S$ such that $s\in S,t\notin S$, and proceeds to generate a dual solution $\mathbf{y}$ as follows: (1) set $y_{v}=1$ if $v\in S$, and $y_{v}=0$ otherwise for each $v\in V$; and (2) set $y_{u,v}=1$ if  $y_{u}=1$ and $y_{v}=0$, and $y_{u,v}=0$ otherwise for each $(u,v)\in E$.  The label assignment $L:V\rightarrow\{0,1\}^*$ constructed by the prover assigns the $1$ bit label $L(v)=y_v$ to each node $v\in V$.
	
	The verifier at each node $v\in V$ deduces the value of the dual variable $y_v$ from the label $L(v)$; and the values $y_{v,u}$ of all the edges incident on $v$ from the label $L(v)$ and the label vector $L^N(v)$. These allow the verifier to verify that $\mathbf{x}$ and $\mathbf{y}$ are feasible primal and dual solutions that satisfy the complementary slackness conditions.
	
	The correctness of this PLS follows directly from the max-flow min-cut theorem and the complementary slackness property. To complete our proof, we note that the sequential runtime of both the prover and verifier is polynomial.
\end{proof}
\begin{theorem}\label{flow-dpls}
	There exists a sequentially efficient DPLS for maximum flow with a proof size of $O(\log n+\log W)$.
\end{theorem}
\begin{proof}
	We obtain a DPLS by means of the VCA method. The prover constructs a minimum $(s,t)$-cut $S$ such that $s\in S$ and $t\notin S$, and uses it as a decomposable 1-approximation for maximum flow. For each node $v\in V$, the sub-label $L_\mathtt{approx}(v)$ is a bit indicating whether $v\in S$. The decomposition function at each node $v\in V$ is simply the sum of capacities of outgoing edges incident on $v$ that cross the cut. To complete the approximation task, the verifier simply verifies that $s\in S$ and $t\notin S$.
\end{proof}

\subsection{Maximum Weight Cut}
\label{section:cut}
Consider a graph $G=(V,E)$ associated with an edge weight function $w:E\rightarrow\{1,\dots W\}$. Given a \emph{cut} $\emptyset\subset S\subset V$ we denote by $E(S)=\{(u,v)\mid u\in S,v\notin S\}$ the set of edges \emph{crossing} it. A \emph{maximum weight cut} is a cut that maximizes $w(S)=\sum_{e\in E(S)}w(e)$.
\begin{theorem}\label{cut-apls}
	There exists a sequentially efficient $2$-APLS for maximum weight cut with a proof size of $1$ bit.
\end{theorem}

\begin{proof}
	Given the cut $S$ derived from the output assignment of the given IO graph $\gio$, the prover assigns to each node $v\in V$ the one bit label $L(v)=1$ if $v\in S$; and $L(v)=0$ otherwise. For each node $v\in V$, let $w(v)=\sum_{(u,v)\in E(S)}w(u,v)$. The verifier at node $v$ uses the label $L(v)$ and the label vector $L^N(v)$ to verify that $w(v)\geq 1/2\cdot \sum_{u\in N(v)}w(u,v)$.
	
	To establish correctness we first show that if $S$ is optimal, then the verifier accepts $\gio$. Assume towards a contradiction that the verifier rejects $\gio$, i.e., there exists a node $v\in V$ such that $w(v)< 1/2\cdot \sum_{u\in N(v)}w(u,v)$. Let us assume w.l.o.g.\  that $v\in S$ and let $\widetilde{S}=S\setminus\{v\}$. Clearly, $w(\widetilde{S})>w(S)$ which contradicts the optimality of $S$. On the other hand, if the verifier accepts $\gio$, then every node $v\in V$ satisfies $w(v)\geq 1/2\cdot \sum_{u\in N(v)}w(v,u)$, which implies that $w(S)\geq 1/2\cdot \sum_{e\in E}w(e)$ and thus $S$ is a 2-approximation to a maximum weight cut.
\end{proof}
\begin{theorem}\label{cut-adpls}
	There exists a sequentially efficient $2$-ADPLS for maximum weight cut with a proof size of $O(\log n+\log W)$.
\end{theorem}
\begin{proof}
	For any graph there exists a cut $S$ satisfying $w(S)\geq 1/2\cdot \sum_{e\in E}w(e)$. Thus, in order to satisfy the correctness requirements, it suffices for the verifier to verify that the given parameter $k$ satisfies $k\geq1/2\cdot \sum_{e\in E}w(e)$. This can be done using the $(-h,-k)$-comparison scheme presented in Section \ref{section:preliminaries}, where $h(\msfi(v),\bot)=\sum_{u\in N(v)}w(v,u)$ for each node $v\in V$.
\end{proof}
\clearpage
\begin{table}
	\begin{center}
		\begin{tabular}{| l| l| l|}
			\hline
			Term &  Abbreviation & Reference\\
			\hline
			input-output graph & IO graph& Section \ref{section:model}\\
			\hline
			distributed graph problem & DGP& Section \ref{section:model}\\
			\hline
			distributed graph minimization problem & MinDGP& Section \ref{section:model}\\
			\hline
			distributed graph maximization problem & MaxDGP& Section \ref{section:model}\\
			\hline
			distributed graph optimization problem & OptDGP& Section \ref{section:model}\\
			\hline
			gap proof labeling scheme & GPLS& Section \ref{section:pls}\\
			\hline
			proof labeling scheme & PLS& Section \ref{section:pls}\\
			\hline
			approximate proof labeling scheme & APLS& Section \ref{section:pls}\\
			\hline
			decision proof labeling scheme & DPLS& Section \ref{section:pls}\\
			\hline
			approximate decision proof labeling scheme & ADPLS& Section \ref{section:pls}\\
			\hline
			verifiable centralized approximation & VCA & Section \ref{section:centralized-verifiable-approximation-method}\\
			\hline
		\end{tabular}
	\end{center}
	\caption{A list of abbreviations.}
	\label{abbreviations}
\end{table}
\begin{table}
	\begin{center}
		\begin{tabular}{| l| l| l| l| l|}
			\hline
			Problem & Graph family & $\alpha$ & Proof size & References \\
			\hline
			minimum edge cover & odd-girth $\geq 2\kappa+1$ & $(\kappa+1)/\kappa$&$O
			(\log\kappa)$ & Theorem \ref{EC-odd-girth} \\
			& odd-girth $\geq 2\kappa+1$ & $(\kappa+1)/\kappa$&$\Omega
			(\log\kappa)$ & Theorem \ref{bound_for_apls_ec}\\
			& odd rings & 1 &$O (\log n)$ &Lemma \ref{lemma-for-bound-ec}\\
			& odd rings & 1 &$\Omega (\log n)$ & Theorem \ref{bound_for_pls_ec}\\
			& bipartite & 1 &1 &Theorem \ref{bipartite-ec}\\
			\hline
			maximum $b$-matching & odd-girth $\geq 2\kappa+1$ & $(\kappa+1)/\kappa$&$O(\log\kappa)$ &Theorem \ref{bmatching-odd-girth}\\
			& odd-girth $\geq 2\kappa+1$ & $(\kappa+1)/\kappa$&$\Omega (\log\kappa)$ &Theorem \ref{bmatching-bound}\\
			& bipartite & 1 &1 &Theorem \ref{bipartite-bmatching}\\
			\hline
			min-weight vertex cover & any & 2&$O(\log n+\log W)$&Lemma \ref{lemma-from-adpls-to-apls}\\
			\hline
			min-weight dominating set & any & $\ln(n)+1$&$O(\log n+\log W)$&Lemma \ref{lemma-from-adpls-to-apls}\\
			\hline
			traveling salesperson & metric & 2&$O(\log n+\log W)$&Lemma \ref{lemma-from-adpls-to-apls}\\
			\hline
			minimum Steiner tree& metric & 2&$O(\log n+\log W)$&Lemma \ref{lemma-from-adpls-to-apls}\\
			\hline
			maximum flow& directed & 1 &1&Theorem \ref{flow-pls}\\
			\hline
			max-weight cut& any & 2&1 &Theorem \ref{cut-apls}\\
			\hline
		\end{tabular}
	\end{center}
	\caption{$\alpha$-APLS results and proof sizes.}
	\label{apls-table}
\end{table}

\begin{table}
	\begin{center}
		\begin{tabular}{| l| l| l| l| l|}
			\hline
			Problem & Graph family & $\alpha$ & Proof size & References \\
			\hline
			minimum edge cover & odd-girth $\geq 2\kappa+1$ & $(\kappa+1)/\kappa$&$O(\log n)$ & Lemma \ref{lemma-adpls-primal-dual}\\
			& odd rings & 1 &$O(\log n)$ &Theorem \ref{ec-ring-dpls}\\
			& bipartite & 1 &$O(\log n)$ &Lemma \ref{lemma-adpls-primal-dual}\\
			\hline
			maximum $b$-matching & odd-girth $\geq 2\kappa+1$ & $(\kappa+1)/\kappa$&$O (\log n+\log W)$ &Lemma \ref{lemma-adpls-primal-dual}\\
			& bipartite & 1 &$O (\log n+\log W)$ &Lemma \ref{lemma-adpls-primal-dual}\\
			\hline
			min-weight vertex cover & any & 2&$O(\log n+\log W)$&Theorem \ref{vc-adpls}\\
			\hline
			min-weight dominating set & any & $\ln(n)+1$&$O(\log n+\log W)$&Theorem \ref{ds-adpls}\\
			\hline
			traveling salesperson & metric & 2&$O(\log n+\log W)$&Theorem \ref{tsp-adpls}\\
			\hline
			minimum Steiner tree& metric & 2&$O(\log n+\log W)$&Theorem \ref{steiner-adpls}\\
			\hline
			maximum flow& directed & 1 &$O (\log n+\log W)$&Theorem \ref{flow-dpls}\\
			\hline
			max-weight cut& any & 2&$O (\log n+\log W)$&Theorem \ref{cut-adpls}\\
			\hline
		\end{tabular}
	\end{center}
	\caption{$\alpha$-ADPLS results and proof sizes.}
	\label{adpls-table}
\end{table}

\begin{table}
	\begin{center}
		\begin{tabular}{| l| l| l|}
			\hline
			Problem &  Inapproximability & Inapproximability w.\ UGC \\
			\hline
			min-weight vertex cover& $1.36-\varepsilon$ \cite{Vertex-Cover-P-NP-Dinur05onthe}&$2-\varepsilon$ \cite{Vertex-Cover-UGC}\\
			\hline
			min-weight dominating set& $(1-o(1))\cdot\ln n$ \cite{Set-Cover-Dinur-2013}&\\
			\hline
			metric traveling salesperson& $123/122-\varepsilon$ \cite{TSP-2013}&\\
			\hline
			metric minimum Steiner tree&$96/95-\varepsilon$ \cite{Steiner-CHLEBIK2008207}&\\
			\hline
			max-weight cut& $1.063-\varepsilon$ \cite{Max-Cut-P-NP}&$1.139-\varepsilon$ \cite{Max-Cut-UGC}\\
			\hline
		\end{tabular}
	\end{center}
	\caption{Known inapproximability bounds with and without the unique games conjecture.}
	\label{inapproximiability}
\end{table}

\clearpage
\bibliographystyle{alpha}

\bibliography{references}

\end{document}